\documentclass[11pt,a4paper,reqno]{amsart}
\usepackage[draft=false,colorlinks,citecolor=blue,linktocpage,breaklinks,hypertexnames=false,pdfpagelabels]{hyperref}
\usepackage[top=3cm,bottom=3cm,left=3cm,right=3cm,headsep=10pt,a4paper]{geometry}

\usepackage[utf8]{inputenc}
\usepackage{amsmath,amssymb,amsthm,dsfont}
\usepackage{physics}
\usepackage{cleveref}
\usepackage{graphicx}
\usepackage{xcolor}
\usepackage{color,soul}

\RequirePackage[framemethod=default]{mdframed}

\newmdenv[skipabove=7pt,
skipbelow=7pt,
backgroundcolor=darkblue!15,
innerleftmargin=5pt,
innerrightmargin=5pt,
innertopmargin=5pt,
leftmargin=0cm,
rightmargin=0cm,
innerbottommargin=5pt,
linewidth=1pt]{tBox}

\newmdenv[skipabove=7pt,
skipbelow=7pt,
backgroundcolor=purple!15,
innerleftmargin=5pt,
innerrightmargin=5pt,
innertopmargin=5pt,
leftmargin=0cm,
rightmargin=0cm,
innerbottommargin=5pt,
linewidth=1pt]{sBox}

\newmdenv[skipabove=7pt,
skipbelow=7pt,
backgroundcolor=darkkblue!15,
innerleftmargin=5pt,
innerrightmargin=5pt,
innertopmargin=5pt,
leftmargin=0cm,
rightmargin=0cm,
innerbottommargin=5pt,
linewidth=1pt]{aBox}

\newmdenv[skipabove=7pt,
skipbelow=7pt,
backgroundcolor=blue2!25,
innerleftmargin=5pt,
innerrightmargin=5pt,
innertopmargin=5pt,
leftmargin=0cm,
rightmargin=0cm,
innerbottommargin=5pt,
linewidth=1pt]{dBox}

\definecolor{darkblue}{RGB}{51,125,255}
\definecolor{darkkblue}{RGB}{51,255,230}
\definecolor{purple}{RGB}{187,51,255}
\definecolor{blue2}{RGB}{51,233,255}

\theoremstyle{plain}
\newtheorem{lemma}{Lemma}[]
\newtheorem{thm}[lemma]{Theorem}
\newtheorem{stp2}{Step}
\newtheorem{stp}{Step}
\newtheorem{assump}{Assumption}
\newtheorem{quest}{Question}
\newtheorem{prop}{Proposition}
\newtheorem{cor}{Corollary}
\newtheorem*{thm*}{Theorem}
\newtheorem*{assump*}{$\highlight{\text{\textbf{Assumption 1}}}$}
\newtheorem*{assump2*}{$\highlight{\text{\textbf{Assumption 2}}}$}
\theoremstyle{definition}
\newtheorem{definition}{Definition}
\theoremstyle{remark}
\newtheorem{remark}{Remark}

\newenvironment{theorem}{\begin{tBox}\begin{thm}}{\end{thm}\end{tBox}}
\newenvironment{defi}{\begin{dBox}\begin{definition}}{\end{definition}\end{dBox}}
\newenvironment{corollary}{\begin{tBox}\begin{cor}}{\end{cor}\end{tBox}}
\newenvironment{step}{\begin{sBox}\begin{stp}}{\end{stp}\end{sBox}}

\newenvironment{assumption}{\begin{aBox}\begin{assump}}{\end{assump}\end{aBox}}
\newenvironment{question}{\begin{aBox}\begin{quest}}{\end{quest}\end{aBox}}

\newcommand{\identity}{\ensuremath{\mathds{1}}}
\newcommand{\hs}{\ensuremath{\mathcal H}}
\newcommand{\ent}[2]{D(#1||#2)}
\newcommand{\entA}[3]{D_{#1}(#2||#3)}

\newcommand{\Z}{\mathbb{Z}}
\newcommand{\N}{\mathbb{N}}

\newcommand{\E}{\mathbb{E}}

\newcommand{\LL}{\mathcal{L}}

\newcommand{\BB}{\mathcal{B}}

\newcommand{\A}{\mathcal{A}}
\newcommand{\SSS}{\mathcal{S}}
\newcommand{\ds}{\displaystyle}
\newcommand{\rla}{\rho_\Lambda}
\newcommand{\sla}{\sigma_\Lambda}
\def\ov{\overset}
\def\un{\underset}

\newcommand{\highlight}[1]{#1}

\renewcommand{\hl}[1]{#1}

\title[On the MLSI for the heat-bath dynamics for 1D systems]{On the modified logarithmic Sobolev inequality for the heat-bath dynamics for 1D systems}

\author[Bardet]{Ivan Bardet}
\author[Capel]{Ángela Capel}
\author[Lucia]{Angelo Lucia}
\author[Pérez-García]{David Pérez-García}
\author[Rouzé]{Cambyse Rouzé}

\address[Bardet]{Inria Paris, France}
\email{bardetivan@gmail.com}

\address[Capel]{Instituto de Ciencias Matemáticas (CSIC-UAM-UC3M-UCM), C/ Nicolás Cabrera 13-15, Campus de Cantoblanco, 28049 Madrid, Spain, Department of Mathematics, Technische Universit\"at M\"unchen, 85748 Garching, Germany, and Munich Center for Quantum Science and Technology (MCQST), München, Germany}
\email{angela.capel@ma.tum.de}

\address[Lucia]{California Institute of Technology, 1200 E California Blvd, MC 305-16, Pasadena, CA 91125, Unites States of America, and Departamento de Análisis Matemático y Matemática Aplicada, Universidad Complutense de Madrid, 28040 Madrid, Spain}

\email{anglucia@ucm.es}

\address[Pérez-García]{Departamento de Análisis Matemático y Matemática Aplicada, Universidad Complutense de Madrid, 28040 Madrid, Spain and Instituto de Ciencias Matemáticas (CSIC-UAM-UC3M-UCM), C/ Nicolás Cabrera 13-15, Campus de Cantoblanco, 28049 Madrid, Spain}

\email{dperezga@ucm.es}

\address[Rouzé]{Department of Mathematics, Technische Universit\"at M\"unchen, 85748 Garching, Germany, and Munich Center for Quantum Science and Technology (MCQST), München, Germany}
\email{rouzecambyse@gmail.com}

\date{\today}

\begin{document}
\maketitle

\begin{abstract}
The mixing time of Markovian dissipative evolutions of open quantum many-body systems can be bounded using optimal constants of certain quantum functional inequalities, such as the modified logarithmic Sobolev constant. For classical spin systems, the positivity of such constants follows from a mixing condition for the Gibbs measure, via quasi-factorization results for the entropy. Inspired by the classical case, we present a strategy to derive the positivity of the modified logarithmic Sobolev constant associated to the dynamics of certain quantum systems from some clustering conditions on the Gibbs state of a local, commuting Hamiltonian. In particular we show that for the heat-bath dynamics for 1D systems, the  modified logarithmic Sobolev constant is positive under the assumptions of a mixing condition on the Gibbs state and a strong quasi-factorization of the relative entropy.
\end{abstract}

\tableofcontents

\section{Introduction}\label{sec:intro}

The fields of quantum information theory and quantum many-body systems have strong ties, as do their classical analogues. In the past years, we have come to see that results and tools developed in quantum information theory have helped to solve fundamental problems in condensed matter physics, whereas some new models created for many-body systems have been used for the storage and transmission of quantum information. There are numerous connections between these two fields, as well as interesting problems lying in their intersection. 

One of these problems is the \textit{thermalizalization} of a quantum system. It has recently generated great interest in both communities for several reasons, one of them being the uprising number of tools  available from quantum information theory \cite{RieraGogolinEisert-Thermalization-2012} \cite{MuellerAdlamMasanesWiebe-Thermalization-2015} to address both the conditions under which a system thermalizes in the infinite time limit, as well as how fast this thermalization occurs. Another one is the possible use of such systems for the implementation of quantum memory devices, as was suggested in the theoretical proposal of \textit{dissipative state engineering} \cite{VerstraeteWolfCirac-StateEngineering-2009} \cite{KrausBuchlerDiehlKantianMicheliZoller-PreparationEntangledStates-2013}, where the authors proposed that quantum dissipative evolutions might constitute a robust way of constructing interesting quantum systems which preserve coherence for relatively long periods. Some experimental results have confirmed this idea, which has also raised the interest in this kind of systems and, thus, in the problem of thermalization.

This paper concerns the question mentioned above, namely how fast a dissipative system thermalizes. This ``velocity'' of thermalization will be studied by means of the \textit{mixing time}, i.e., the time that it takes for every initial state undergoing a dissipative evolution to be almost indistinguishable from the thermal equilibrium state. In particular, we will be interested in physical systems for which this convergence is fast enough, in a regime that is called \textit{rapid mixing}. Several bounds for the mixing time, and thus, conditions for rapid mixing to hold, can be found via the optimal constants associated to some quantum functional inequalities, as it is the case for the \textit{spectral gap} (optimal constant for the \textit{Poincaré inequality}) \cite{[TKRWV10]} or the \textit{modified logarithmic Sobolev constant} (associated to the \textit{modified logarithmic Sobolev inequality}) \cite{[KT13]}. We will focus on the latter.

This problem was previously addressed in the classical setting. In \cite{[D02]}, the authors showed that a classical spin system in a lattice, for a certain dynamics and a clustering condition in the Gibbs measure associated to this dynamics, satisfies a modified logarithmic Sobolev inequality. In \cite{[Cesi01]}, the usual logarithmic Sobolev inequality was studied via another similar condition of clustering in the Gibbs measure. Both results were inspired by the seminal work of Martinelli and Olivieri \cite{MartinelliOlivieri-ApproachEquilibriumGlauberI-1994} \cite{[Mar99]} and aimed to notably simplify their proof via a result of \textit{quasi-factorization of the entropy}\footnote{This and other terms in italics used in the Introduction will be defined precisely later in the paper, in Section \ref{sec:preliminaries}.} in terms of some conditional entropies. Previously, a result of \textit{quasi-factorization of the variance}  \cite{BertiniCancriniCesi-ClassicalSpectralGap-2002} had been used to prove positivity of the spectral gap for certain dynamics, under certain conditions in the Gibbs measure. 

The latter found its quantum analogue in \cite{[KB14]}, where the authors introduced the notion of \textit{conditional spectral gap} and proved positivity of the spectral gap for the Davies and heat-bath dynamics associated to a local commuting Hamiltonian, via a result of quasi-factorization of the variance, under a condition of \textit{strong clustering} in the Gibbs state. In this paper, we aim to study the quantum analogue of the aforementioned strategy to obtain positivity of a modified logarithmic Sobolev inequality via a result of quasi-factorization of the relative entropy, thus obtaining  an exponential improvement on the mixing time with respect to the spectral gap case.

The main purpose of this paper is to present a  strategy to obtain a positive modified logarithmic Sobolev constant and apply it to the particular case of the heat-bath dynamics in 1D. Our strategy is based on the following five points:

\begin{enumerate}
\item Definition of a \textit{conditional modified logarithmic Sobolev constant} (conditional MLSI constant, for short).
\item Definition of the \textit{clustering condition} on the Gibbs state.
\item Quasi-factorization of the relative entropy in terms of a \textit{conditional relative entropy} (introduced in \cite{[CLP18a]}).
\item Recursive geometric argument to provide a lower bound for the global modified logarithmic Sobolev constant (MLSI constant, for short) in terms of the conditional MLSI constant in a fixed-sized region.
\item Positivity of the conditional MLSI constant.
\end{enumerate}

Remark that point (1) is not needed in the classical case\hl{, i.e. whenever all the involved density matrices are diagonal in a product basis of the Hilbert space of states}. Indeed, there the Dobrushin–Lanford–Ruelle (DLR) condition \cite{[Dob68]},  \cite{[LR69]} allows to reduce the study of the MLSI on a domain and its boundary. In the quantum case, the DLR condition fails~\cite{fannes1995boundary} and the definition of the conditional MLSI is here to compensate this feature, by taking into account the whole exterior of the domain, not only its boundary. This leads to several new difficulties, as illustrated by point (5) which is completely trivial in the classical case.

These five points together will imply the positivity of the modified logarithmic Sobolev constant independently of the system size. Note that the first two consist of careful definitions of some notions, the first of which will be used during the proof, whereas the second one will constitute the initial assumptions to impose on the Gibbs state. The third point has been previously addressed in \cite{[CLP18a]}. It constitutes the quantum analogue of the results of quasi-factorization of the entropy in the classical case, and provides an upper bound for the relative entropy between two states in terms of the sum of two conditional relative entropies and a multiplicative error term that measures how far the second state is from a tensor product. This result allows  to reduce from the global MLSI constant to a conditional one in the fourth point. Finally, in the last point, we prove that the conditional MLSI constant is positive and independent of the system size, yielding thus the same fact for the global one. In general, this will be the trickiest part of the strategy. 

\hl{This strategy of five steps constitutes a generalization of the one used in $\text{\cite{[CLP18a]}}$ to prove that the heat-bath dynamics with tensor product fixed point has a positive MLSI constant, lower bounded by $1/2$ (see also $\text{\cite{beigi2018quantum}}$). This result showed, in particular, that local depolarizing on-site noise destroys very fast any quantum memory. Then, it is only natural that in the current manuscript we weaken the conditions imposed on the fixed point and address the analogous problem for a more general class of quantum systems evolving under this dynamics. Indeed}
, in the main result of this paper, we apply \hl{the}  strategy \hl{above} to the heat-bath dynamics in dimension 1, to obtain positivity of the MLSI constant under two additional assumptions on the Gibbs state.

\hl{Before proceeding to the informal exposition of the main result of this manuscript, let us emphasize the importance of the strategy to prove MLSI presented in this paper, which is also at the core of the thesis $\text{\cite{Capel-Thesis-2019}}$. Indeed, after the completion of the first version of this manuscript, a modification of this strategy for a new conditional relative entropy has been used in $\text{\cite{BardetCapelRouze-ApproximateTensorization-2020}}$ and $\text{\cite{CapelRouzeStilckFranca-MLSIcommuting-2020}}$ to obtain the first examples of positivity of the MLSI for quantum lattice spin systems independently of the system size, solving thus a long-standing open problem. However, those examples concern a different class of semigroups, called Schmidt generators, whose main feature is that they present better properties regarding the decomposition of the algebras over different regions of the lattice. Therefore, while the current manuscript does not solve the problem of positivity of MLSI for the heat-bath dynamics, the strategy we present here has already proven to be successful in analyizing MLSI, even beyond the specific generators considered here.}

\subsection{Informal exposition of the main result}

In this article, we consider the quantum heat-bath dynamics in 1D, whose generator is constructed following the same idea than for the clasical heat-bath Monte-Carlo algorithm. More specifically, given a finite lattice $\Lambda \subset \Z$ and a state $\rho_\Lambda \in \SSS_\Lambda$, it is defined as:
\begin{equation*}
\LL^*_\Lambda (\rho_\Lambda) = \un{x \in \Lambda}{\sum} \left( \sigma^{1/2}_\Lambda \sigma^{-1/2}_{x^c} \rho_{x^c} \sigma^{-1/2}_{x^c} \sigma^{1/2}_\Lambda - \rho_\Lambda  \right),
\end{equation*}
where the first term in the sum of the RHS coincides with the Petz recovery map for the partial trace at every site $x \in \Lambda$, composed with the partial trace in $x$, and $\sigma_\Lambda$ is the Gibbs state of a commuting $k$-local Hamiltonian. 

Analogously to the classical case and the quantum spectral gap case, as part of the previous strategy we need to assume that a couple of clustering conditions on the Gibbs state hold. The first one is related to the exponential decay of correlations in the Gibbs state of the given commuting Hamiltonian and is satisfied, for example, by Gibbs states at high-enough temperature depending logarithmically of the system size. 

\begin{assump*}[Mixing condition]
Let $C, D \subset \Lambda$ be the union of non-overlapping finite-sized segments of $\Lambda$. The following inequality holds for positive constants $K_1, K_2$ independent of $\Lambda$:
\begin{equation*}
\norm{\sigma_C^{-1/2} \otimes \sigma_D^{-1/2} \,  \sigma_{CD} \,  \sigma_C^{-1/2} \otimes \sigma_D^{-1/2} - \identity_{CD}}_\infty \leq K_1 \operatorname{e}^{- K_2 \text{d}(C, D)},
\end{equation*}
where $\text{d}(C, D)$ is the distance between $C$ and $D$, i.e., the minimum distance between two segments of $C$ and $D$\footnote{The precise choice of the distance measure will be introduced later.}.
\end{assump*}

The second assumption constitutes a stronger form of quasi-factorization of the relative entropy than the ones appearing in \cite{[CLP18a]}. An example where it holds is for Gibbs states verifying $\sigma_\Lambda = \un{x \in \Lambda}{\bigotimes} \sigma_x$. 

\begin{assump2*}[Strong quasi-factorization]
 Given $X \subset \Lambda$, for every $\rho_\Lambda \in \SSS_\Lambda$ the following inequality holds
\begin{equation*}
\entA X \rla \sla \leq f_X(\sigma_{\Lambda}) \un{x \in X}{\sum} \entA x \rla \sla,
\end{equation*}
where  $1 \leq f_X(\sigma_{\Lambda}) < \infty$ depends only on $\sigma_\Lambda$ and does not depend on the size of $\Lambda$, whereas $D_X(\rho_\Lambda || \sigma_\Lambda)$, resp. $\entA x \rla \sla$, is the conditional relative entropy in $X$, resp. $x$, of $\rho_\Lambda$ and $\sigma_\Lambda$, whose explicit form will be recalled in Subsection \ref{subsec:entropy}.
\end{assump2*}

Then, the main result of this paper is the positivity of the modified logarithmic Sobolev constant for the heat-bath dynamics in 1D.

\begin{thm*}
Considering that the two previous assumptions hold, the modified logarithmic Sobolev constant of the generator associated to the  heat-bath dynamics in 1D systems with invariant state the Gibbs state of a local commuting Hamiltonian is strictly positive and independent of $\abs{\Lambda}$.
\end{thm*}

The rigorous statement of this theorem will be given in Theorem \ref{thm:mainresult}.

\subsection{Layout of the paper}

In Section \ref{sec:preliminaries}, we introduce the necessary notation, preliminary notions and basic properties to follow the rest of the paper. In Section \ref{sec:tools} we prove several technical tools (which do not depend on the necessary assumptions for the main result), of independent interest, that will be of use in the proof of the main result, which we subsequently address in Section \ref{sec:MLSI-HB-1D}, providing a complete and self-contained proof. In Section \ref{sec:assumptions}, we discuss the assumptions imposed on the Gibbs state, providing examples and situations in which they hold. We finally conclude in Section \ref{sec:conclusions} with some open problems. 

\section{Preliminaries and notation}\label{sec:preliminaries}

\subsection{Notation}\label{subsec:notation}

In this paper we consider finite dimensional Hilbert spaces. For $\Lambda$ a set of $\abs{\Lambda}$ parties, we denote the multipartite finite dimensional Hilbert space of $\abs{\Lambda}$ parties by $\hs_\Lambda = \underset{x \in \Lambda}{\bigotimes} \hs_x$, whose dimension is $d_\Lambda$. Throughout this text, $\Lambda$ will often consist of 3 parties, and we will denote by $\hs_{ABC}= \hs_A \otimes \hs_B \otimes \hs_C$ the corresponding tripartite Hilbert space. Furthermore, most of the paper concerns quantum spin lattice systems and we often assume that $\Lambda \subset \subset \Z^d$ is a finite subset. In general, we use uppercase Latin letters to denote systems or sets.

For every finite dimensional $\hs_\Lambda$, we denote the associated set of bounded linear operators by $\BB_\Lambda := \BB(\hs_\Lambda)$, and by $\A_\Lambda := \A(\hs_\Lambda)$ its subset of observables, i.e. Hermitian operators, which we denote by lowercase Latin letters. We further denote by $\SSS_\Lambda := \SSS(\hs_\Lambda)= \qty{f_\Lambda \in \A_\Lambda \, : \, f_\Lambda \geq 0 \text{ and }\tr[f_\Lambda]=1}$ the set of density matrices, or states, and denote its elements by lowercase Greek letters. In particular, whenever they appear in the text, Gibbs states are denoted by $\sigma_\Lambda$. We usually denote the space where each operator is defined using the same subindex as for the space, but we might drop it when it is unnecessary. 

In this manuscript we often consider quantum channels, i.e. completely positive and trace-preserving maps. In general, a linear map $\mathcal T:\BB_\Lambda \rightarrow \BB_\Lambda$ is called a superoperator.  We write $\identity$ for the identity matrix and $\text{id}$ for the identity superoperator. For bipartite spaces $\hs_{AB}= \hs_A \otimes \hs_B$, we consider the natural inclusion $\A_A \hookrightarrow \A_{AB}$ by identifying each operator $f_A \in \A_A$ with $f_{A} \otimes \identity_B$. In this way, we define the modified partial trace in $A$ of $f_{AB} \in \A_{AB}$ by $\highlight{\tr_B[f_{AB}]} \otimes \identity_B$, but we denote it by $\highlight{\tr_B[f_{AB}]} $ in a slight abuse of notation.  Moreover, we say that an operator $g_{AB} \in \A_{AB}$ has support in $A$ if it can be written as $g_A \otimes \identity_B$ for some operator $g_A \in \A_A$. Note that given $f_{AB} \in \A_{AB}$, we write $f_A:= \tr_B[f_{AB}]$.

Finally,  given $x,y \in \Lambda \subset \subset \Z^d$, we denote by $d(x,y)$ the Euclidean distance between $x$ and $y$ in $\Z^d$. Hence, the distance between two subsets of $\Lambda$, $A$ and $B$, is given by $d(A,B):= \text{min} \qty{d(x,y) : x \in A, y \in B}$. Furthermore, we denote by $\norm{\cdot}_\infty$ the usual operator norm, as well as by $\norm{\cdot}_1=\tr[\abs{\cdot}] $ the trace-norm.




\subsection{Entropies}\label{subsec:entropy}

\subsubsection{Von Neumann entropy}\label{subsubsec:vNentropy}

Let $\hs_\Lambda$ be a finite dimensional Hilbert space and consider $\rho_\Lambda \in \SSS_\Lambda$. The \textit{von Neumann entropy} of $\rho_\Lambda$ is defined as
\begin{equation*}
S(\rho_\Lambda):= - \tr[\rho_\Lambda \log \rho_\Lambda].
\end{equation*}

The applications of this notion to quantum statistical mechanics and quantum information theory are numerous. Here we focus on one of its most fundamental properties, which will appear often throughout the text.

\begin{prop}[Strong subadditivity, \cite{LiebRuskai-Subadditivity-1973}]
Let $\hs_{ABC} = \hs_A \otimes \hs_B \otimes \hs_C$ be a tripartite Hilbert space and consider $\rho_{ABC} \in \SSS_{ABC}$. Then, the following inequality holds:
\begin{equation*}
S(\rho_{ABC}) + S(\rho_B) \leq S(\rho_{AB}) + S(\rho_{BC}).
\end{equation*}

\end{prop}

\subsubsection{Relative entropy}\label{subsubsec:RE}
A measure of distinguishability between two states that will appear often throughout this text is the relative entropy. Let $\hs_\Lambda$ be a finite dimensional Hilbert space and consider $\rho_\Lambda, \sigma_\Lambda \in \SSS_\Lambda$. We define the \textit{relative entropy} of $\rho_\Lambda$ and $\sigma_\Lambda$ as
\begin{equation*}
D(\rho_\Lambda || \sigma_\Lambda):=  \tr[\rho_\Lambda (\log \rho_\Lambda - \log \sigma_\Lambda)].
\end{equation*}

Some fundamental properties of the relative entropy that will be of use are the following. 

\begin{prop}[Properties of the relative entropy, \cite{Wehrl-Entropy-1978}]
Let $\hs_{AB} = \hs_A \otimes \hs_B $ be a bipartite Hilbert space and consider $\rho_{AB}, \sigma_{AB} \in \SSS_{AB}$. Then, the following properties hold:
\begin{enumerate}
\item \textbf{Non-negativity.} $D(\rho_{AB} || \sigma_{AB}) \geq 0$ and $D(\rho_{AB}|| \sigma_{AB}) = 0$ if, and only if, $\rho_{AB} = \sigma_{AB}$.
\item \textbf{Additivity.} $ D(\rho_A \otimes \rho_B || \sigma_A \otimes \sigma_B )= D(\rho_A || \sigma_A) + D(\rho_B || \sigma_B)$.
\item \textbf{Superadditivity.} $ D(\rho_{AB}  || \sigma_A \otimes \sigma_B ) \geq D(\rho_A || \sigma_A) + D(\rho_B || \sigma_B)$.
\item \textbf{Data processing inequality.} For every quantum channel $\mathcal T: \SSS_{AB} \rightarrow \SSS_{AB}$, $D(\rho_{AB} || \sigma_{AB}) \geq D(\mathcal T(\rho_{AB}) ||  \mathcal T(\sigma_{AB}))$.
\end{enumerate}

\end{prop}

\subsubsection{Conditional relative entropy}\label{subsubsec:CRE}

The conditional relative entropy provides the value of the distinguishability between two states in a certain system given the value of their distinguishability in a subsystem. Let $\hs_{AB}= \hs_A \otimes \hs_B$ be a bipartite finite dimensional Hilbert space and consider $\rho_{AB}, \sigma_{AB} \in \SSS_{AB}$. The conditional relative entropy in $A$ of $\rho_{AB}$ and $\sigma_{AB}$ is given by
\begin{equation*}
D_A(\rho_{AB} || \sigma_{AB}) := D(\rho_{AB} || \sigma_{AB}) - D(\rho_B || \sigma_B).
\end{equation*}

We recall in the next proposition some properties of the conditional relative entropy that will be of use in the next sections. 

\begin{prop}[Some properties of the conditional relative entropy, \cite{[CLP18a]}]
Let $\hs_{AB}= \hs_A \otimes \hs_B$ be a bipartite finite dimensional Hilbert space and consider $\rho_{AB}, \sigma_{AB} \in \SSS_{AB}$. Then, the following properties hold:
\begin{enumerate}
\item \textbf{Non-negativity.} $D_A (\rho_{AB} || \sigma_{AB}) \geq 0$.
\item $D_A(\rho_{AB}|| \sigma_A \otimes \sigma_B) = I_\rho(A:B) + D(\rho_A || \sigma_A)$,
where $I_\rho(A:B):= D(\rho_{AB} || \rho_A \otimes \rho_B)$ is the mutual information.
\end{enumerate}
\end{prop}

\subsection{Modified logarithmic Sobolev constant}\label{subsec:MLSI}
Given a state $\rho_\Lambda$  in an open quantum many-body system under the Markov approximation, its time evolution is described by a one-paramenter semigroup of completely positive trace-preserving maps $\mathcal{T}_t^* := \operatorname{e}^{t \mathcal{L}_\Lambda^*}$, also known as \textit{quantum Markov semigroup} (QMS), where $\mathcal{L}_\Lambda^*: \SSS_\Lambda \rightarrow \SSS_\Lambda$ denotes the generator of the semigroup, which is called \textit{Liouvillian} or \textit{Lindbladian}, since its dual version in the Heisenberg picture satisfies the Lindblad (or GKLS) form \cite{[L76a]}, \cite{GKS-GKLSform-1976} for every $X_\Lambda \in \BB_\Lambda:$
\begin{equation*}
\LL_\Lambda(X_\Lambda)= i [H,X_\Lambda] + \frac{1}{2} \un{k=1}{\ov{l}{\sum}} \left[ 2 L_k^* X_\Lambda L_k - (L_k^* L_k X_\Lambda + X_\Lambda L_k^* L_k)  \right],
\end{equation*}
where $H \in \A_\Lambda$, the $L_k \in \BB_\Lambda$ are the \textit{Lindblad operators} and $[\cdot, \cdot]$ denotes the commutator. 

We say that the QMS is \textit{primitive} if there is a unique full-rank $\sigma_\Lambda \in \SSS_\Lambda$ which is invariant for the generator, i.e. such that $\LL_\Lambda^*(\sigma_\Lambda)=0$. Furthermore, we say that the Lindbladian is \textit{reversible}, or satisfies the \textit{detailed balance condition}, with respect to a state $\sigma_\Lambda \in \SSS_\Lambda$ if its version for observables verifies
\begin{equation*}
\left\langle f_\Lambda, \LL_\Lambda (g_\Lambda) \right\rangle_{\sigma_\Lambda} =  \left\langle \LL_\Lambda  (f_\Lambda), g_\Lambda \right\rangle_{\sigma_\Lambda}
\end{equation*}
for every $f_\Lambda, g_\Lambda \in \A_\Lambda$, where this weighted scalar product is defined for every $f_\Lambda, g_\Lambda \in \A_\Lambda$ by 
\begin{equation*}
\left\langle f_\Lambda, g_\Lambda \right\rangle_{\sigma_\Lambda} := \tr[f_\Lambda \sigma_\Lambda^{1/2} g_\Lambda \sigma_\Lambda^{1/2}].
\end{equation*}

We recall now the notion of entropy production as the derivative of the relative entropy in the following form. 

\begin{definition}\label{def:EntropyProduction}
Let $\Lambda \subset \subset \Z^d$ be a finite lattice and let $\hs_\Lambda$ be the associated Hilbert space. Let $\LL_\Lambda^* : \SSS_\Lambda \rightarrow \SSS_\Lambda$ be a primitive reversible Lindbladian with fixed point $\sigma_\Lambda\in \SSS_\Lambda$. Then, for every $\rho_\Lambda \in \SSS_\Lambda$, the \textit{entropy production} is defined as
\begin{equation*}
\operatorname{EP}(\rho_\Lambda):= - \left. \frac{\operatorname{d}}{\operatorname{dt}}\right\vert_{t=0} D(\rho_t || \sigma_\Lambda) = - \tr[\LL_\Lambda^*(\rho_\Lambda)(\log \rho_\Lambda - \log \sigma_\Lambda)],
\end{equation*}
where we are writing $\rho_t := \mathcal{T}^*_t(\rho_\Lambda)$.

\end{definition}

Note that the entropy production of a primitive QMS only vanishes on $\sigma_\Lambda$. The fact that both the negative derivative of the relative entropy between the elements of the semigroup and the fixed point and the relative entropy between the same states have the same kernel and converge to zero in the long time limit, for every possible initial state for the semigroup, allows us to consider the possibility of bounding one in terms of the other. This is the reason to define a modified logarithmic Sobolev inequality and its optimal constant. 

\begin{definition}\label{def:MLSI}
Let $\Lambda \subset \subset \Z^d$ be a finite lattice, $\hs_\Lambda$ its associated Hilbert space and $\LL_\Lambda^* : \SSS_\Lambda \rightarrow \SSS_\Lambda$ a primitive reversible Lindbladian with fixed point $\sigma_\Lambda\in \SSS_\Lambda$. Then, the \textit{modified logarithmic Sobolev constant} (MLSI constant) is defined as 
\begin{equation*}
\alpha(\LL_\Lambda^* ):= \un{\rho_\Lambda \in S_\Lambda}{\text{inf}} \frac{- \tr[\LL_\Lambda^*(\rho_\Lambda)(\log \rho_\Lambda - \log \sigma_\Lambda)]}{2 D(\rho_\Lambda || \sigma_\Lambda)}.
\end{equation*}
\end{definition}

A family of quantum logarithmic Sobolev inequalities was introduced in \cite{[KT13]}, where the modified logarithmic Sobolev inequality, whose optimal constant we have just recalled, is identified with the $1$-logarithmic Sobolev inequality. In the same paper, it is shown that the existence of a positive MLSI constant implies a bound in the mixing time of an evolution, i.e., the time that it takes for every initial state to be almost indistinguishable of the fixed point, which constitutes an exponential improvement in terms of the system size to the bound provided by the existence of a positive spectral gap. Indeed, if $\alpha(\LL_\Lambda^* )>0$, then for every $\rho_\Lambda \in \SSS_\Lambda$:
\begin{equation*}
\norm{\rho_t - \sigma_\Lambda}_1 \leq \sqrt{2 \log(\norm{\sigma_\Lambda^{-1}}_\infty)} \operatorname{e}^{-\alpha(\LL_\Lambda^*) t}.
\end{equation*}

This constitutes a way to obtain sufficient conditions for a QMS to satisfy rapid mixing, a property that has profound implications in the system, such as stability against external perturbations  \cite{Cubitt2015} and the fact that its fixed point satisfies an area law for the mutual information \cite{Brandao2015}.

\subsection{Gibbs states}\label{gibbs}

Given a finite lattice $\Lambda \subset \subset \Z^d$, let us define a $k$-local bounded potential as $\Phi: \Lambda \rightarrow \A_\Lambda$ such that, for any $x \in \Lambda$, $\Phi(x)$ is a Hermitian matrix supported in a ball of radius $k$ centered at $x$ and there exists a constant $C< \infty$ such that $\norm{\Phi(x)}_\infty <C$ for every $x\in \Lambda$. 

 We define the \textit{Hamiltonian} from this potential in the following way: For every subset $A \subset \Lambda$, the Hamiltonian in $A$, $H_A$, is given by
\begin{equation*}
H_A := \un{x\in A}{\sum} \Phi(x).
\end{equation*}

We further say that this potential is \textit{commuting} if $[\Phi(x), \Phi(y)]=0$ for every $x,y \in \Lambda$.

Consider now $A \subset \Lambda$ and $\Phi$ a bounded $k$-local potential. Since the potential is local, we can define the boundary of $A$ as
\begin{equation*}
\partial A := \qty{x \in \Lambda\setminus A \, | \, d(x, A)< k }
\end{equation*}
and we denote by $A \partial$ the union of $A$ and its boundary. Note that $H_A$ clearly has support in $A \partial$. Since in this paper we only focus on 1D systems, for every bounded connected subset $A\subset \Lambda$, the boundary will be composed of two parts, which we will intuitively denote by $(\partial A)_\text{Left}$ and  $(\partial A)_\text{Right}$, respectively.

In the full lattice $\Lambda \subset \subset \Z^d$, the Gibbs state is defined as
\begin{equation*}
\sigma_\Lambda := \frac{\operatorname{e}^{-\beta H_\Lambda}}{\tr[\operatorname{e}^{-\beta H_\Lambda}]}.
\end{equation*}

Note that, by a slight abuse of notations, we will denote by $\sigma_A$ for $A \subset \Lambda$ the state given by $\tr_{A^c}[\sigma_\Lambda]$, which should not be confused with the restricted Gibbs state corresponding to the terms of the Hamiltonian $H_A$. 

\subsection{Heat-bath generator}\label{subsec:heat-bath}

Let $\Lambda \subset \subset \Z^d$ be a finite lattice and $\Phi: \Lambda \rightarrow \A_\Lambda$ a $k$-local bounded commuting potential. Consider $\sigma_\Lambda$ to be the associated Gibbs state. Given $A \subseteq \Lambda$, we define the \textit{heat-bath conditional expectation} as follows: for every $\rho_\Lambda \in \SSS_\Lambda$,
\begin{equation*}
\E_A^*(\rho_\Lambda) := \sigma_\Lambda^{1/2} \sigma_{A^c}^{-1/2} \rho_{A^c} \sigma_{A^c}^{-1/2} \sigma_\Lambda^{1/2}. 
\end{equation*}

Note that it is a quantum channel and, moreover, it coincides with the Petz recovery map for the partial trace in $A$ with respect to $\sigma_\Lambda$, composed with the partial trace in $A$ \cite{[Petz86]}, \hl{ i.e. $\mathbb E_A^* (\cdot) := \mathcal P_{\tr_A}^{\sigma_\Lambda} \circ \tr_A [\cdot] $ for}
\begin{equation*}
\highlight{ P_{\tr_A}^{\sigma_\Lambda} (\cdot) := \sigma_\Lambda^{1/2} \sigma_{A^c}^{-1/2} ( \cdot ) \sigma_{A^c}^{-1/2} \sigma_\Lambda^{1/2} \, . }
\end{equation*} 
Furthermore, it is the dual map of  the minimal conditional expectation that appears in \cite{[KB14]}. As opposed to what its name suggests, it is not a usual conditional expectation, but a quasi-conditional expectation \cite{OhyaPetz-Entropy-1993}, since it lacks some of the basic properties in the definition of conditional expectation. 

We can now define the \textit{heat-bath generator on $\Lambda$} by
\begin{equation*}
\LL_\Lambda^* (\rho_\Lambda):= \un{x \in \Lambda}{\sum} \left( \E_x^*(\rho_\Lambda) - \rho_\Lambda \right),
\end{equation*}
for every $\rho_\Lambda \in \SSS_\Lambda$. Analogously for every $A \subset \Lambda$, we denote by $\LL_A^*$ the generator where the summation is only over elements $x\in A$. Note that the Lindbladian is defined as the sum of terms containing conditional expectations considered over single sites. Some basic properties concerning the heat-bath generator are collected in the following proposition.

\begin{prop}[\cite{[KB14]}]
Let $\Lambda \subset \subset \Z^d$ be a finite lattice and $\Phi: \Lambda \rightarrow \A_\Lambda$ a $k$-local bounded commuting potential. Then, the following properties hold:
\begin{enumerate}
\item For any $A\subset \Lambda$, $\LL_A^*$ is the generator of a semigroup of CPTP maps of the form $\operatorname{e}^{t \LL_A^*}$.
\item $\LL_\Lambda^*$ is $k$-local, in the sense that each individual composing term acts non-trivially only on balls of radius $k$.  
\item For any $A,B \subset \Lambda$, we have
\begin{equation*}
\LL_A^*+\LL_B^* = \LL_{A\cup B}^*+ \LL_{A\cap B}^*.
\end{equation*}
\end{enumerate}
\end{prop}

To conclude this subsection, let us introduce two concepts that will be of use in the proof of the main result. They are conditional versions of notions defined on the whole system, in the same spirit as the conditional relative entropy.

\begin{defi}\label{def:conditionalEntropyProduction}
Let $\Lambda \subset \subset \Z^d$ be a finite lattice and let $\LL_\Lambda^* : \SSS_\Lambda \rightarrow \SSS_\Lambda$ be the heat-bath generator with fixed point $\sigma_\Lambda\in \SSS_\Lambda$. Given $A \subset \Lambda$, we define the \textit{entropy production in $A$} for every $\rho_\Lambda \in \SSS_\Lambda$ by
\begin{equation*}
\operatorname{EP}_A(\rho_\Lambda):= - \tr[\LL_A^*(\rho_\Lambda)(\log \rho_\Lambda - \log \sigma_\Lambda)].
\end{equation*}
\end{defi}

Considering the notions of entropy production in a subsystem and conditional relative entropy, one can address again the problem of relating both of them via an inequality, thus obtaining a conditional version of the aforementioned MLSI constant.

\begin{defi}\label{def:conditionalMLSIconstant}
Let $\Lambda \subset \subset \Z^d$ be a finite lattice and let $\LL_\Lambda^* : \SSS_\Lambda \rightarrow \SSS_\Lambda$ be the heat-bath generator with fixed point $\sigma_\Lambda\in \SSS_\Lambda$. Given $A \subset \Lambda$, we define the \textit{conditional MLSI constant} by
\begin{equation*}
\alpha_\Lambda(\LL_A^* ):= \un{\rho_\Lambda \in \SSS_\Lambda}{\text{inf}} \frac{- \tr[\LL_A^*(\rho_\Lambda)(\log \rho_\Lambda - \log \sigma_\Lambda)]}{2 D_A(\rho_\Lambda || \sigma_\Lambda)},
\end{equation*}
where $D_A(\rho_\Lambda || \sigma_\Lambda)$ is the conditional relative entropy introduced in Subsection \ref{subsec:entropy}.
\end{defi}

In the classical setting, there is no need to define a conditional MLSI constant, since it coincides with the MLSI constant due to the DLR condition \cite{[D02]}. Not only this last property fails in general in the quantum case~\cite{fannes1995boundary}, but also the study  of the conditional MLSI constant is essential in our case,  as it is  part of our strategy to prove the positivity of the MLSI constant. 

\subsection{Quantum Markov chains}\label{subsec:QMC}

Consider a tripartite space $\hs_{ABC}= \hs_A \otimes \hs_B \otimes \hs_C$. We define a \textit{recovery map} $\mathcal{R}_{B \rightarrow BC}$ from $B$ to $BC$ as a completely positive trace-preserving map that reconstructs the $C$-part of a state $\sigma_{ABC}\in \SSS_{ABC}$  from its $B$-part only. If that reconstruction is possible, i.e., if for a certain $\sigma_{ABC}\in \SSS_{ABC}$ there exists such  $\mathcal{R}_{B \rightarrow BC}$ verifying 
\begin{equation*}
\sigma_{ABC}= \mathcal{R}_{B \rightarrow BC} (\sigma_{AB}),
\end{equation*}
we say that $\sigma_{ABC}$ is a \textit{quantum Markov chain} (QMC) between $A \leftrightarrow B \leftrightarrow C$. When this is the case, the recovery map can be taken to be the Petz recovery map. \hl{More specifically, $\sigma_{ABC}$ is a QMC($A \leftrightarrow B \leftrightarrow C$) if, and only if, it is a fixed point of the composition of the Petz recovery map for the partial trace in $C$, with respect to the state $\sigma_{BC}$,  with the partial trace in $C$, i.e.}
\begin{equation*}
\sigma_{ABC} = \highlight{\mathcal{P}_{\tr_C}^{\sigma_{BC}}\circ \tr_C[\sigma_{ABC}]  = } \sigma_{BC}^{1/2}\highlight{ \sigma_{B}^{-1/2}} \sigma_{AB}\highlight{ \sigma_{B}^{-1/2}} \sigma_{BC}^{1/2}. 
\end{equation*}

This class of states has been deeply studied in the last years. In the next proposition, we collect an equivalent condition for a state to be a QMC.

\begin{thm}[\cite{[Petz86]}, \cite{Petz-MonotonicityRelativeEntropy-2003}]
Let  $\hs_{ABC}= \hs_A \otimes \hs_B \otimes \hs_C$ be a tripartite Hilbert space and $\sigma_{ABC} \in \SSS_{ABC}$. Then, $\sigma_{ABC}$ is a quantum Markov chain, if, and only if, 
$I_\sigma(A:C | B)=0$,  for $I_\sigma(A:C | B)=S(\sigma_{AB})+S(\sigma_{BC})-S(\sigma_{ABC})-S(\sigma_B)$ the quantum conditional mutual information.
\end{thm}

Another important equivalent condition for a state to be a quantum Markov chain, concerning its structure as a direct sum of tensor products, appears in the next result.

\begin{thm}[Theorem 6 of \cite{[HJPW04]}]\label{thm:StructQMC} A tripartite state $\sigma_{ABC}$ of $\hs_{A}\otimes \hs_B\otimes \hs_C$ satisfies $I_\sigma(A:C|B)=0$ if and only if there exists a decomposition of system $B$ as $\hs_B=\bigoplus_j \hs_{b_j^L}\otimes \hs_{b_j^R}$ into a direct sum of tensor products such that 
	$$\sigma_{ABC}=\bigoplus_j q_j\,\sigma_{Ab_j^L}\otimes \sigma_{b_j^RC},$$
	with the state $\sigma_{Ab_j^L}$ (resp. the state $\sigma_{b_j^RC}$) being on $\hs_{A}\otimes \hs_{b_j^L}$ (resp. on $\hs_{b_j^R}\otimes \hs_{C}$) and a probability distribution $\{q_j\}$.
\end{thm}

Turning now to Gibbs states, as they were introduced in the previous subsection, we recall an important result about their Markovian structure. 
\begin{thm}[Theorem 3 of \cite{[BP12]}]\label{theo1}
	Given a $k$-local commuting potential on $\Lambda$, its associated Gibbs state $\sigma_\Lambda$ is a quantum Markov network, that is for all disjoint subsets $A,B,C\subset\Lambda$ such that $B$ shields $A$ from $C$ with $d(A,C)> k$, $I_\sigma(A:C|B)=0$.
\end{thm}	

\hl{ The notion of `shield' used in the statement of the theorem denotes that the systems $ A$ and  $B $ are not adjacent, i.e. no site of  $A$   is at distance  $1$   from  $C$  and viceversa. Moreover, the condition d$(A,C)>k$   implies that there are at least  $k$ sites of $ B $ between any site of $A$  and another one of $C$}. Therefore, combining the results of Theorem \ref{theo1} and Theorem \ref{thm:StructQMC}, we obtaining the following essential result for the structure of Gibbs states. 

\begin{cor}
Let $\Lambda \subset \subset \Z^d$ be a finite lattice and $\sigma_\Lambda$ the Gibbs state of a  commuting Hamiltonian. Then, for any tripartition $\tilde{A}\tilde{B}\tilde{C}$ of $\Lambda$ such that $\tilde{B}$ shields $\tilde{A}$ from $\tilde{C}$, the state $\sigma_{\Lambda}$ can be decomposed as
 \begin{align}\label{eq2}
 \sigma_\Lambda=\bigoplus_j q_j\,\sigma_{\tilde{A}\tilde{b}_j^L}\otimes \sigma_{\tilde{b}_j^R\tilde{C}}.
 \end{align}
\end{cor}

Using the previous properties for quantum Markov chains, we can easily show the identity of the next proposition.

\begin{prop}\label{prop:identityQMC}
Let $\hs_{ABC}=\hs_A \otimes \hs_B \otimes \hs_C$ be a tripartite Hilbert space and  $\sigma_{ABC}$ a quantum Markov chain between $A \leftrightarrow B \leftrightarrow C$. Then, the following identity holds: 
\begin{equation}\label{eq:logarithms-QMC}
\log \sigma_{ABC} + \log \sigma_B = \log \sigma_{BC} + \log \sigma_{AB}.
\end{equation}
\end{prop}

\begin{proof}
Since $\sigma_{ABC}$ is a quantum Markov chain between $A \leftrightarrow B \leftrightarrow C$, by Theorem \ref{thm:StructQMC}  we can write it as
 \begin{align}
 \sigma_\Lambda=\bigoplus_j q_j\,\sigma_{A b_j^L}\otimes \sigma_{b_j^R C}.
 \end{align}
 
Hence,
\begin{align*}
  - \log \sigma_{ABC} + & \log \sigma_{BC}  + \log \sigma_{AB} - \log \sigma_{B}   \\
& =  \un{j}{\sum} \left(  - \log \sigma_{A b_j^L}\otimes \sigma_{b_j^R C} + \log \sigma_{b_j^L}\otimes \sigma_{b_j^R C} + \log \sigma_{A b_j^L}\otimes \sigma_{b_j^R }- \log \sigma_{ b_j^L}\otimes \sigma_{b_j^R }   \right) \\
 & = 0,
\end{align*} 
where we have used the fact that the logarithm of a tensor product splits as a sum of logarithms.
\end{proof}

As a consequence of this identity, we have the following result.

\begin{cor}
Let $\hs_{ABC}=\hs_A \otimes \hs_B \otimes \hs_C$ be a tripartite Hilbert space and  $\sigma_{ABC}$ a quantum Markov chain between $A \leftrightarrow B \leftrightarrow C$. Then, for any $\rho_{ABC} \in \SSS_{ABC}$, the following identity holds: 
\begin{equation}\label{eq:CRE=CRE+CMI}
\entA A {\rho_{ABC}} {\sigma_{ABC}} = \entA A {\rho_{AB}} {\sigma_{AB}} + I_\rho (A:C |B),
\end{equation}
where $I_\rho (A:C |B)$ denotes the \textit{conditional mutual information} of $\rho_{ABC}$.

In particular,
\begin{equation*}
\entA A {\rho_{ABC}} {\sigma_{ABC}} \geq \entA A {\rho_{AB}} {\sigma_{AB}}.
\end{equation*}

\end{cor}

\begin{proof}
Since $\sigma_{ABC}$ is a quantum Markov chain between $A \leftrightarrow B \leftrightarrow C$, by Proposition \ref{prop:identityQMC} we have
\begin{align*}
\entA A {\rho_{ABC}} {\sigma_{ABC}}  - &\entA A {\rho_{AB}} {\sigma_{AB}}  \\
&=  \ent {\rho_{ABC}} {\sigma_{ABC}} - \ent {\rho_{BC}} {\sigma_{BC}} - \ent {\rho_{AB}} {\sigma_{AB}} + \ent {\rho_{B}} {\sigma_{B}}  \\
& = \underbrace{- S (\rho_{ABC}) + S(\rho_{BC}) + S(\rho_{AB})  - S(\rho_B) }_{I_\rho (A:C |B)} \\
& \; \; \; + \tr[ \rho_{ABC} \left(  - \log \sigma_{ABC} + \log \sigma_{BC}  + \log \sigma_{AB} - \log \sigma_{B}    \right)  ] \\
 & = I_\rho (A:C |B) .
\end{align*} 

In particular, since $ I_\rho (A:C |B) \geq 0$ for every state $\rho_{ABC} \in \SSS_{ABC}$, 
\begin{equation*}
\entA A {\rho_{ABC}} {\sigma_{ABC}} \geq \entA A {\rho_{AB}} {\sigma_{AB}}.
\end{equation*}
\end{proof}


\section{Technical tools}\label{sec:tools}

This section aims at presenting a collection of technical results which will be necessary in the proof of the main result of the paper in Section \ref{sec:MLSI-HB-1D}. Some of them, as we will see below, are of independent interest to quantum information theory. Note that all the results that appear in this section hold independently of Assumptions \ref{assump:1} and \ref{assump:2} and do not depend on the geometry of $\Lambda$.

The main technical result of this section is  Theorem \ref{thm:kernelA=kernelx}. In its proof, we will make use of the following lemma, which provides a lower bound for a conditional entropy production in a single site (see Definition \ref{def:conditionalEntropyProduction}) in terms of a conditional relative entropy in the same single site. 

\begin{lemma}\label{lemma:EP>D}
For a single site $x \in \Lambda$, and for every $\rho_\Lambda, \sigma_\Lambda \in \SSS_\Lambda$, the following holds
\begin{equation}\label{eq:EP>D}
\operatorname{EP}_x(\rho_\Lambda) \geq \entA x {\rho_\Lambda} {\sigma_\Lambda},
\end{equation}
where $\operatorname{EP}_x(\rho_\Lambda) $ is defined with respect to $\sigma_\Lambda$. Therefore, $\operatorname{EP}_A(\rho_\Lambda)\ge 0$ for any $A\subset\Lambda$ and $\rho\in \mathcal{S}_\Lambda$.
\end{lemma}

\begin{proof}
The proof is a direct consequence of the data processing inequality and the fact that $\E_x^*(\cdot)$ is the Petz recovery map for the partial trace in $x$, composed with the partial trace (and, in particular, a quantum channel). Indeed, let us recall that $\text{EP}_x(\rho_\Lambda)$ is given by
\begin{align}\label{eq:lemmaEP>D-1}
\text{EP}_x(\rho_\Lambda)& = - \tr[ \LL_x^*(\rho_\Lambda) (\log \rho_\Lambda - \log \sigma_\Lambda) ] \nonumber \\
& = \tr[ \left( \rho_\Lambda - \E_x^*(\rho_\Lambda)  \right)  (\log \rho_\Lambda - \log \sigma_\Lambda) ] \nonumber \\
&= \ent {\rho_\Lambda} {\sigma_\Lambda} - \tr[ \E_x^*(\rho_\Lambda) (\log \rho_\Lambda - \log \sigma_\Lambda)  ] .
\end{align}

In the second term of (\ref{eq:lemmaEP>D-1}), let us add and substract $\log \E_x^*(\rho_\Lambda)$. Then,
\begin{align}\label{eq:lemmaEP>D-2}
 \tr[ \E_x^*(\rho_\Lambda) (\log \rho_\Lambda - \log \sigma_\Lambda)  ]  &=   \tr[ \E_x^*(\rho_\Lambda) (\log \rho_\Lambda - \log \sigma_\Lambda +\log \E_x^*(\rho_\Lambda) - \log \E_x^*(\rho_\Lambda) )  ]  \nonumber \\
 &= - \ent {\E_x^*(\rho_\Lambda) } {\rho_\Lambda}+  \ent {\E_x^*(\rho_\Lambda) } {\sigma_\Lambda} \nonumber \\
 & \leq  \ent {\E_x^*(\rho_\Lambda) } {\sigma_\Lambda},
\end{align}
where we have used the fact that the relative entropy of two states is always non-negative.

Finally, since $\E_x^*(\cdot)$ is the Petz recovery map for the partial trace in $x$ composed with the partial trace (denote $\E_x^*(\cdot) = \highlight{\mathcal{P}^\sigma_{\tr_x}} \circ \tr_x[ \cdot ] $), note that $\sigma_\Lambda$ is a fixed point. Then,
\begin{align*}
\ent {\E_x^*(\rho_\Lambda) } {\sigma_\Lambda} &= \ent {\highlight{\mathcal{P}^\sigma_{\tr_x}}  \circ \tr_x[\rho_\Lambda] } {\highlight{\mathcal{P}^\sigma_{\tr_x}}  \circ \tr_x[\sigma_\Lambda] } \\
& \leq \ent {\rho_{x^c}} {\sigma_{x^c}},
\end{align*}
and thus
\begin{equation*}
\text{EP}_x(\rho_\Lambda) \geq  \ent {\rho_\Lambda} {\sigma_\Lambda} - \ent {\rho_{x^c}} {\sigma_{x^c}} = \entA x {\rho_\Lambda} {\sigma_\Lambda}.
\end{equation*}

\end{proof}

\begin{remark}
If we recall the definition for conditional MLSI constant introduced in the previous section, Lemma \ref{lemma:EP>D} can clearly be seen as a lower bound for the conditional MLSI constant at a single site $x \in \Lambda$ for the heat-bath dynamics, i.e.,
\begin{equation*}
\alpha_\Lambda(\LL_x^* ) \geq \frac{1}{2}.
\end{equation*}

This inequality, in particular, can be used to prove positivity of the MLSI constant for the heat-bath dynamics when $\sigma_\Lambda$ is a tensor product, as it appears in \cite{[CLP18a]} (see also \cite{beigi2018quantum} and \cite{Bardet-NonCommFunctInequalities-2017}).
\end{remark}

\begin{remark}
Note that, in the previous lemma, we have only used the fact that the partial trace is a quantum channel and  $\E_x^*(\cdot)$ its Petz recovery map composed with it. Hence, in more generality, Lemma \ref{lemma:EP>D} could be stated as: Let $\mathcal{T}$ be a quantum channel and denote by $\widehat{\mathcal{T}}$ its Petz recovery map with respect to $\sigma_\Lambda$. Then, for any $\rho_\Lambda \in \SSS_\Lambda$ the following holds
\begin{equation}
\tr[ (\rho_\Lambda - \widehat{\mathcal{T}} \circ \mathcal{T} (\rho_\Lambda)) \left(\log  \rho_\Lambda - \log \sigma_\Lambda \right) ] \geq \ent {\rho_\Lambda} {\sigma_\Lambda} - \ent { \mathcal{T}  (\rho_\Lambda)} { \mathcal{T}  (\sigma_\Lambda)}.
\end{equation}
\end{remark}

Another tool that will be of use in the main result of this section is the following lemma, which appeared first in \cite{[MOZ98]}. It can be seen as an equivalence between blocks of spins, and allows us to prove an equivalence between the usual conditional Lindbladian associated to the heat-bath dynamics in $A \subseteq \Lambda$, given as a sum of local terms, and a modified one given as a unique term. Note that it is stated in the Heisenberg picture.  

\begin{lemma}[\cite{[MOZ98]}]\label{lemma:Zegarlinski}
Let $A \subseteq \Lambda$ and let $\sigma_\Lambda$ be the Gibbs state of the $k$-local commuting Hamiltonian mentioned above. There exist constants $0<c_A, C_A < \infty$, possibly depending on $A$ but not on $\Lambda$, such that for any $f_\Lambda \in \A_\Lambda $ the following holds:
\begin{equation}\label{eq:lemma-blocks}
c_A \un{x \in A}{\sum} \left\langle f_\Lambda , f_\Lambda - \E_x(f_\Lambda)   \right\rangle_{\sigma_\Lambda} \leq  \left\langle f_\Lambda , f_\Lambda - \E_A(f_\Lambda)   \right\rangle_{\sigma_\Lambda} \leq C_A  \un{x \in A}{\sum}  \left\langle f_\Lambda , f_\Lambda - \E_x (f_\Lambda)   \right\rangle_{\sigma_\Lambda}, 
\end{equation}
where $\E_x$, resp. $\E_A$, is the dual of $\E_x^*$, resp. of $\E_A^*$, and is given by
\begin{equation*}
\E_x(f_\Lambda):= \sigma_{x^c}^{-1/2} \tr_x[\sigma_\Lambda^{1/2} f_\Lambda \sigma_\Lambda^{1/2} ]\sigma_{x^c}^{-1/2}, 
\end{equation*}
for every $f_\Lambda \in \A_\Lambda $ and analogously for $\E_A$.
\end{lemma}

Let us now state and prove the main technical result of this section, which will be essential for the proof of Theorem \ref{thm:mainresult}, but has independent interest on its own.

\begin{theorem}\label{thm:kernelA=kernelx}
Let $\Lambda \subset \subset \Z^d$ be a finite lattice and let $\sigma_\Lambda \in \SSS_\Lambda$ be the Gibbs state of a commuting Hamiltonian  over $\Lambda$. For any $A \subseteq \Lambda$ and $\rho_\Lambda \in \SSS_\Lambda$, the following equivalence holds:
\begin{equation}
\rho_\Lambda = \E_A^* (\rho_\Lambda)  \Leftrightarrow \rho_\Lambda = \E_x^* (\rho_\Lambda) \; \; \forall x \in A.
\end{equation}
\end{theorem}

\begin{proof}
Let us first recall that, for every $\rho_\Lambda \in \SSS_\Lambda$, the local Lindbladian in $A\subseteq \Lambda$ is given by
\begin{equation*}
\LL_A^*(\rho_\Lambda) = \un{x \in A}{\sum} \left( \E_x^* (\rho_\Lambda)  - \rho_\Lambda  \right),
\end{equation*}
and define
\begin{equation*}
\widetilde{\LL}_A^*(\rho_\Lambda) =\E_A^* (\rho_\Lambda)  - \rho_\Lambda.
\end{equation*}

Analogously, defining the superoperator $\Gamma_{\sigma_\Lambda}:f_\Lambda\mapsto\sigma_\Lambda^{1/2}f_\Lambda\sigma_{\Lambda}^{1/2} $, we can write every observable $f_\Lambda \in \A_\Lambda$ as
\begin{equation}\label{eq:Gamma}
f_\Lambda = \Gamma_{\sigma_\Lambda}^{-1} (\rho_\Lambda)= \sigma_\Lambda^{-1/2} \rho_\Lambda \sigma_\Lambda^{-1/2},
\end{equation} 
and thus we have 
\begin{equation*}
\LL_A(f_\Lambda) =  \un{x \in A}{\sum} \left( \E_x (f_\Lambda)  - f_\Lambda  \right), 
\end{equation*}
\begin{equation*}
\widetilde{\LL}_A(f_\Lambda) =\E_A (f_\Lambda)  - f_\Lambda.
\end{equation*}

 With this notation, inequality (\ref{eq:lemma-blocks}) in Lemma \ref{lemma:Zegarlinski} can be rewritten as 
\begin{equation*}
- c_A  \left\langle f_\Lambda , \LL_A(f_\Lambda)  \right\rangle_{\sigma_\Lambda} \leq -  \left\langle f_\Lambda , \widetilde{\LL}_A(f_\Lambda)   \right\rangle_{\sigma_\Lambda} \leq - C_A    \left\langle f_\Lambda , \LL_A(f_\Lambda)  \right\rangle_{\sigma_\Lambda} ,
\end{equation*} 
and thus,
\begin{equation*}
\left\langle f_\Lambda , \LL_A(f_\Lambda)  \right\rangle_{\sigma_\Lambda} = 0 \quad\Leftrightarrow\quad\forall x\in A\,,\quad\left\langle f_\Lambda , \LL_x(f_\Lambda)  \right\rangle_{\sigma_\Lambda} = 0 \quad\Leftrightarrow\quad  \left\langle f_\Lambda , \widetilde{\LL}_A(f_\Lambda)   \right\rangle_{\sigma_\Lambda} = 0,
\end{equation*}
which thanks to the detailed-balance condition, \hl{and the subsequent positivity of the generators}, leads to
\begin{equation}\label{eq:equivLA}
\LL_A(f_\Lambda)= 0  \quad\Leftrightarrow\quad\forall x\in A\,,\quad\LL_x(f_\Lambda)= 0  \quad\Leftrightarrow\quad \widetilde{\LL}_A(f_\Lambda)=0.    
\end{equation}

Now, because of (\ref{eq:Gamma}), one can easily see that $\mathbb{E}_x^*=\Gamma_{\sigma_\Lambda}\circ \mathbb{E}_x\circ\Gamma_{\sigma_\Lambda}^{-1}$ and the same holds for $\mathbb{E}_A^*$. Hence, (\ref{eq:equivLA}) is equivalent to
\begin{equation}\label{eq:equivLA*}
\LL_A^*(\rho_\Lambda)= 0 \quad\Leftrightarrow\quad \forall x\in A\,,\quad \LL_x^*(\rho_\Lambda)=0 \quad\Leftrightarrow\quad \widetilde{\LL}^*_A(\rho_\Lambda)=0.
\end{equation}

Recalling the expressions for $\LL_A^*(\rho_\Lambda)$ and $\widetilde{\LL}^*_A(\rho_\Lambda)$, we obtain:
\begin{equation*}
\rho_\Lambda = \E_A^* (\rho_\Lambda) \; \;  \Leftrightarrow \; \;  \rho_\Lambda = \E_x^* (\rho_\Lambda) \; \; \; \forall x \in A. 
\end{equation*}

\end{proof}

This result can also be stated in terms of conditional relative entropies. Indeed, note that, as a consequence of Petz's characterization for conditions of equality in the data processing inequality, all the conditions above can be seen as necessary and sufficient conditions for vanishing conditional relative entropies. We have then the following corollary.

\begin{corollary}\label{cor:kernelA=kernelx}
Let $\Lambda \subset \subset \Z^d$ be a finite quantum lattice and let $\sigma_\Lambda \in \SSS_\Lambda$ be the Gibbs state of a commuting Hamiltonian. For any $A \subseteq \Lambda$ and $\rho_\Lambda \in \SSS_\Lambda$, the following equivalence holds:
\begin{equation}
\entA A {\rla}  \sla = 0  \Leftrightarrow \entA x \rla \sla = 0 \; \; \forall x \in A.
\end{equation}
\end{corollary}

Another consequence of the previous result is that a state is recoverable from a certain region whenever it is recoverable from several components of that region that cover it completely, no matter the size of those components. More specifically, we have the following corollary.

\begin{corollary}\label{cor:kernelAB=kernelA}
Given a finite lattice $\Lambda$, a partition of it into three subregions $A,B,C$, and $\sigma_{ABC}$ the Gibbs state of a commuting Hamiltonian, if we denote by $\E_A^* (\cdot)$ the conditional expectation on $A$ associated to the heat-bath dynamics (with respect to the Gibbs state), we have for any $\rho_{ABC}\in \SSS_{ABC}$:
\begin{equation}\label{eq:kernelAB=kernelA}
\E_{AB}^*(\rho_{ABC})= \rho_{ABC} \Leftrightarrow \left\lbrace  \begin{array}{c}
 \E_{A}^*(\rho_{ABC})= \rho_{ABC} \\
 \text{and} \\
  \E_{B}^*(\rho_{ABC})= \rho_{ABC} .
\end{array}   \right.  
\end{equation}

In particular,
\begin{equation}
\entA {AB} {\rho_{ABC}} {\sigma_{ABC}}= 0 \Leftrightarrow \left\lbrace  \begin{array}{c}
\entA {A} {\rho_{ABC}} {\sigma_{ABC}}= 0 \\
 \text{and} \\
\entA {B} {\rho_{ABC}} {\sigma_{ABC}}= 0 .
\end{array}   \right.   
\end{equation}

\end{corollary}

\begin{proof}
By virtue of Theorem \ref{thm:kernelA=kernelx}, it is clear that
\begin{align*}
\E_{AB}^*(\rho_{ABC})= \rho_{ABC} & \Leftrightarrow  \E_x^*(\rho_{ABC})= \rho_{ABC} \; \;  \forall x \in A \cup B \\
 &\Leftrightarrow \left\lbrace  \begin{array}{ccc}
  \E_x^*(\rho_{ABC})= \rho_{ABC}  \; \; \forall x \in A  & \Leftrightarrow & \E_{A}^*(\rho_{ABC})= \rho_{ABC} \\
  \text{and}& & \text{and}\\
    \E_x^*(\rho_{ABC})= \rho_{ABC}  \; \; \forall x \in B  & \Leftrightarrow & \E_{B}^*(\rho_{ABC})= \rho_{ABC} 
\end{array}   \right. 
\end{align*}

The second part is a direct consequence of \cite{[Petz86]} and Corollary \ref{cor:kernelA=kernelx}.
\end{proof}

\section{Main result}\label{sec:MLSI-HB-1D}

In this section, we state and prove the main result of this paper, namely a static sufficient condition on the Gibbs state of a $k$-local commuting Hamiltonian for the heat-bath dynamics in 1D to have a positive modified logarithmic-Sobolev constant (MLSI constant in short). For that, we first need to introduce two assumptions that need to be considered in order to prove the result, and which will be discussed in further detail in the next section. 

The first condition can be interpreted as an exponential decay of correlations in the Gibbs state of the commuting Hamiltonian. In Section \ref{subsec:assump1} we will see that only a weaker assumption is necessary, although this form is preferable here for its close connections to its classical analogue \cite{[D02]}.

\begin{assumption}[Mixing condition]\label{assump:1}
Let $\Lambda \subset \subset \Z$ be a finite chain and let $C, D \subset \Lambda$ be the union of non-overlapping finite-sized segments of $\Lambda$. Let  $\sla$ be the Gibbs state of a commuting Hamiltonian. The following inequality holds for certain positive constants $K_1, K_2$ independent of $\Lambda, C, D$:

\begin{equation*}
\norm{\sigma_C^{-1/2} \otimes \sigma_D^{-1/2} \,  \sigma_{CD} \,  \sigma_C^{-1/2} \otimes \sigma_D^{-1/2} - \identity_{CD}}_\infty \leq K_1 \operatorname{e}^{- K_2 \text{d}(C, \, D)},
\end{equation*}
where $\text{d}(C, D)$ is the distance between $C$ and $D$, i.e., the minimum distance between two segments of $C$ and $D$.
\end{assumption}

The second condition that needs to be assumed constitutes a strong form of quasi-factorization of the relative entropy. 

\begin{assumption}[Strong quasi-factorization]\label{assump:2}
Let $\Lambda \subset \subset \Z$ be a finite chain and $X \subset \Lambda$.  Let $\sla$ be the Gibbs state of a $k$-local commuting Hamiltonian. For every $\rho_\Lambda \in \SSS_\Lambda$, the following inequality holds
\begin{equation}
\entA X \rla \sla \leq f_X(\sigma_\Lambda) \un{x \in X}{\sum} \entA x \rla \sla,
\end{equation}
where $1 \leq f_X(\sigma_\Lambda) < \infty$ depends only on $\sigma_\Lambda$ on \hl{ $X \partial$ (in particular, it is independent of $\abs{\Lambda}$)}. 
\end{assumption}

This form of quasi-factorization is stronger than the one that appeared in \cite{[CLP18a]}, since another conditional relative entropy appears in the LHS of the inequality, instead of a relative entropy as in the main results of quasi-factorization of the aforementioned paper. Moreover, the error term depends only on the second state, as in usual quasi-factorization results, but only on its value in the regions where the relative entropies are being conditioned and their boundaries. In particular, it is independent of the size of the chain.

As in the case of Assumption \ref{assump:1}, we will see in Subsection \ref{subsec:assump2} that only a weaker condition is necessary for Theorem \ref{thm:mainresult} to hold true, since this condition will only appear in the proof concerning sets $X$ of small size. 

Let us now state and prove  the main result of this paper, namely the positivity of the MLSI constant for the heat-bath dynamics in 1D.

\begin{theorem}\label{thm:mainresult}
Let $\Lambda \subset \subset \Z$ be a finite chain. Let $\Phi : \Lambda \rightarrow \A_\Lambda$ be a $k$-local commuting potential, $H_\Lambda = \un{x \in \Lambda}{\sum} \Phi (x) $ its corresponding  Hamiltonian, and denote by $\sigma_\Lambda$ its Gibbs state. Let $\LL_\Lambda^*$ be the generator of  the heat-bath dynamics. Then, if Assumptions \ref{assump:1} and \ref{assump:2} hold, the MLSI constant of $\LL_\Lambda^*$ is strictly positive and independent of $\abs{\Lambda}$.
 
\end{theorem}

The proof of this result will be split into four parts. First, we need to define a splitting of the chain into two (not connected) subsets $A,B \subset \Lambda$, with a certain geometry so that 1) they cover the whole chain, 2) their intersection is large enough and 3) each one of them is composed of smaller segments of fixed size, but large enough to contain two non-overlapping half-boundaries of two other segments, respectively.

\begin{figure}
\begin{center}
\includegraphics[scale=0.32]{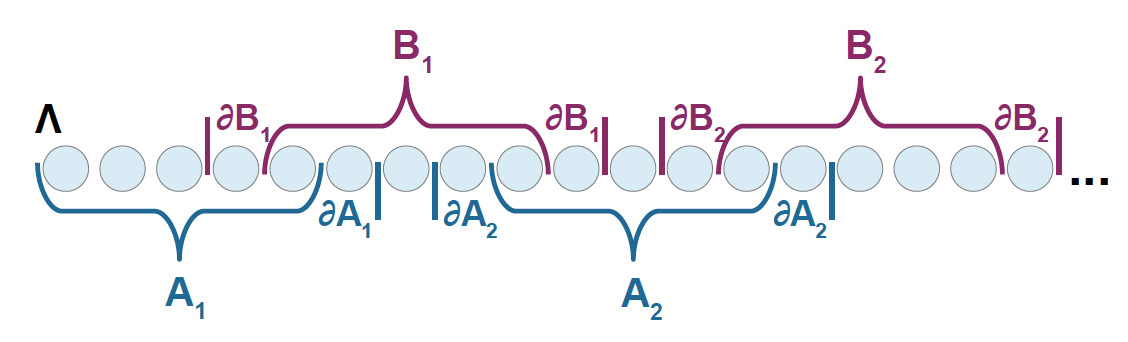}
\end{center}
\caption{Splitting of $\Lambda$ in fixed-sized subsets $A_i$ and $B_i$, of which we just show the first four terms. We reduce for simplicity to the case $k=2, l=1$.}
\label{fig:step1}
\end{figure}

More specifically, fix $l \in \N$ so that $\ds K_1 e^{-K_2 l} < \frac{1}{2}$, for $K_1$ and $K_2$ the constants appearing in the mixing condition, and consider the splitting of $\Lambda$ given in terms of $A$ and $B$ verifying the following conditions (see Figure \ref{fig:step1}):
\begin{enumerate}
\item $\Lambda= A \cup B$.

\item   $\ds A=\un{i=1}{\ov{n}{\bigcup}} A_i$ and $\ds B=\un{j=1}{\ov{n}{\bigcup}} B_j$.

\item $\abs{A_i \cap B_i}= \abs{B_i \cap A_{i+1}}=l$ for every $i=1, \ldots, n-1$.

\item $\abs{A_i}= \abs{B_j} = 2 (k+l)- 1$ for all $i,j=1, \ldots , n$, where $k$ comes from the $k$-locality of the Hamiltonian.
\end{enumerate}

Note that the total size of $\Lambda$ is then $n(4k+2l-2)+l$ sites. Hence, fixing $l$ and $k$ as already mentioned, we can restrict our study here to lattices of size $n(4k+2l-2)+l$  for every $n \in \N$, as we will be interested in the scaling properties in the limit.

In the first step, considering this decomposition of the chain, we show an upper bound for the relative entropy of two states on $\Lambda$ (the second of them being the Gibbs state) in terms of the sum of two conditional relative entropies in $A$ and $B$, respectively, and a multiplicative error term that measures how far the reduced state $\sigma_{A^c B^c}$ is from a tensor product between $A^c$ and $B^c$, where $A^c:= \Lambda \setminus A$ and $B^c:=\Lambda \setminus B$.

\begin{step}\label{step:1}

For the regions $A$ and $B$ defined above, and for any $\rho_\Lambda \in \SSS_\Lambda$, we have
\begin{equation}\label{eq:step1}
\ent {\rho_\Lambda} {\sigma_\Lambda} \leq \frac{1}{1-2 \norm{h(\sigma_{A^c B^c})}_\infty} \left[ \entA A {\rho_\Lambda} {\sigma_\Lambda} +  \entA B {\rho_\Lambda} {\sigma_\Lambda}   \right],
\end{equation}
where 
\begin{equation*}
h(\sigma_{A^c B^c} )= \sigma_{A^c}^{-1/2} \otimes \sigma_{B^c}^{-1/2} \,  \sigma_{A^c B^c} \,  \sigma_{A^c}^{-1/2} \otimes \sigma_{B^c}^{-1/2} - \identity_{A^c B^c}.
\end{equation*}
\end{step}

\begin{proof}
In \cite{[CLP18]}, the authors showed that, given a bipartite space $\hs_{XY}= \hs_X \otimes \hs_Y$, for every $\rho_{XY}, \sigma_{XY} \in \SSS_{XY}$ one has 
\begin{equation}\label{eq:step1-1}
(1+ 2 \norm{h(\sigma_{XY})}_\infty) \ent {\rho_{XY}} {\sigma_{XY}} \geq \ent {\rho_{X}} {\sigma_{X}}  + \ent {\rho_{Y}} {\sigma_{Y}}, 
\end{equation}
for $h(\sigma_{XY})$ given by
\begin{equation*}
h(\sigma_{XY} )= \sigma_X^{-1/2} \otimes \sigma_Y^{-1/2} \,  \sigma_{XY} \,  \sigma_X^{-1/2} \otimes \sigma_Y^{-1/2} - \identity_{XY}.
\end{equation*}

Considering a tripartite space $\hs_{XZY}$ and $\rho_{XZY}, \sigma_{XZY} \in \SSS_{XZY}$, inequality (\ref{eq:step1-1}) was shown to be equivalent to 
\begin{equation*}
(1+ 2 \norm{h(\sigma_{XY})}_\infty) \ent {\rho_{XZY}} {\sigma_{XZY}} \geq \ent {\rho_{X}} {\sigma_{X}}  + \ent {\rho_{Y}} {\sigma_{Y}}
\end{equation*}
in \cite{[CLP18a]}, and recalling the definition for the conditional relative entropy in $XZ$ and $ZY$, respectively, this inequality can be rewritten as
\begin{equation*}
\ent {\rho_{XZY}} {\sigma_{XZY}} \leq \frac{1}{1-2 \norm{h(\sigma_{XY})}_\infty} \left[ \entA {XZ} {\rho_{XZY}} {\sigma_{XZY}} +  \entA {ZY} {\rho_{XZY}} {\sigma_{XZY}}  \right].
\end{equation*}

Finally, inequality (\ref{eq:step1}) follows just by replacing in this expression $\Lambda=XZY, A \cap B = Z, A^c = Y$ and $B^c= X$. 
\end{proof}

For the second step of the proof, we focus on one of the two components of $\Lambda$, e.g. $A$, and upper bound the conditional relative entropy of two states in the whole $A$ in terms of the sum of the conditional relative entropies in its fixed-size small components. In this case, there is no multiplicative error term, due to the structure of quantum Markov chain of the Gibbs state between one component, its boundary, and the complement, and the fact that the boundaries of these components do not overlap.

\begin{step}\label{step:2}
For $\ds A=\un{i=1}{\ov{n}{\bigcup}} A_i$ defined as above (see Figure \ref{fig:step2}), and for every $\rho_\Lambda \in \SSS_\Lambda$, the following holds:
\begin{equation}\label{eq:step2}
\entA A {\rho_\Lambda} {\sigma_\Lambda} \leq \un{i=1}{\ov{n}{\sum}} \entA {A_i} {\rho_\Lambda} {\sigma_\Lambda}. 
\end{equation}

\end{step}

\begin{figure}
	\begin{center}
		\includegraphics[scale=0.3]{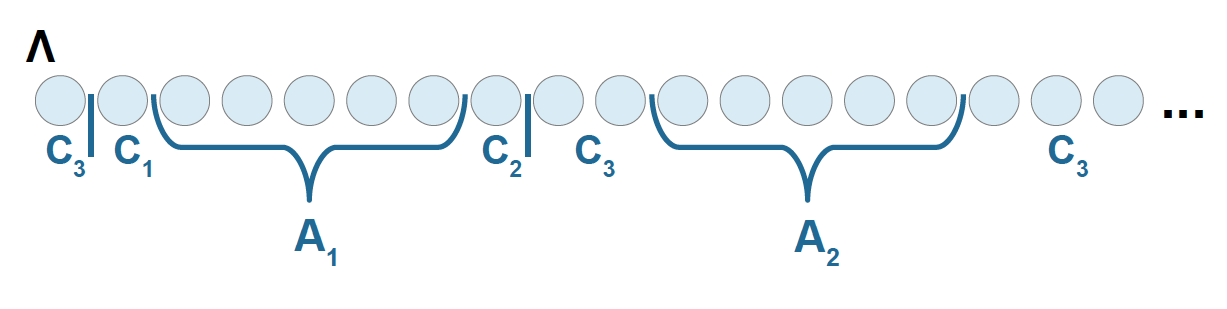}
	\end{center}
	\caption{Splitting of $A$ in fixed-sized subsets $A_i$ so that their boundaries do not overlap. For simplicity we restrict to the case $k=2, l=1$.}
	\label{fig:step2}
\end{figure}

\begin{proof}
	Without loss of generality, we assume that $A=A_1\cup A_2$ (the general result follows by induction in the number of subsets $A_i$). For convenience, we denote $ \ent {\rho_A} {\sigma_A}$, resp. $\entA A {\rho_\Lambda} {\sigma_\Lambda}$, by $D(A)$, resp. $D_A(\Lambda)$, since we are considering the same states $\rho_\Lambda$ and $\sigma_\Lambda$ in every (conditional) relative entropy. 
	
	With this notation, it is enough to show:
\begin{equation}\label{eq:step2-1}
D_A(\Lambda)-D_{A_1}(\Lambda)-D_{A_2}(\Lambda) \leq 0.
\end{equation}

	In this step, we are going to consider three different regions: $A_1, \partial A_1 $ and $\left( A_{1 } \partial \right)^c$, but we could argue analogously for $A_2$. We write $C_1:= \left( \partial A_1 \right)_{\text{Left}}$ and $C_2:= \left( \partial A_1 \right)_{\text{Right}}$, the left and right components of $\partial A_1$, respectively, and $C_3$ all the sites that appear left to $C_1$, between $C_2$ and $A_2$, and right to $A_2$, i.e., $C_3:=(A_1\partial \cup A_2)^c$. We further write $C:= \un{m=1}{\ov{3}{\bigcup}} C_m$. Note that $A_1$, $A_2$ and $C$ are disjoint and that $A_1 \cup A_2 \cup C= \Lambda$. Then, 
	\begin{align}
	 D_A(\Lambda)-D_{A_1}(\Lambda)-D_{A_2}(\Lambda)&=D(\Lambda)-D(C)-D(\Lambda)+D(A_2\cup C)-D(\Lambda)+D(A_1\cup C)\nonumber\\
	 &= -D(C)-D(\Lambda)+D(A_2\cup C)+D(A_1\cup C)\nonumber\\
	 &=\tr\rho_\Lambda(-\log\rho_\Lambda-\log\rho_C+\log\rho_{A_1C}+\log\rho_{A_2C})\nonumber\\
	 & \, \, \, \,+\tr\rho_\Lambda(\log\sigma_C-\log\sigma_{A_2C}+\log\sigma_{\Lambda}-\log\sigma_{A_1C}) \nonumber\\
	 &=S(\rho_\Lambda)+S(\rho_C)-S(\rho_{A_1C})-S(\rho_{A_2C})\nonumber\\
	 &\, \, \, \,+\tr\rho_\Lambda(\log\sigma_C-\log\sigma_{A_2C}+\log\sigma_{\Lambda}-\log\sigma_{A_1C})\nonumber \\
	 &\le\tr\rho_\Lambda(\log\sigma_C-\log\sigma_{A_2C}+\log\sigma_{\Lambda}-\log\sigma_{A_1C}),\label{eq3}
	\end{align} 
	where the last inequality follows from strong subadditivity of the von Neumann entropy. Now, from the structure of quantum Markov chain of the Gibbs state and by Proposition \ref{prop:identityQMC}, the sum of logarithms vanishes. 
	

\end{proof}

Combining expressions (\ref{eq:step1}) and (\ref{eq:step2}) from Steps \ref{step:1} and \ref{step:2}, respectively, we get
\begin{equation}\label{eq:step3-1}
\ent {\rho_\Lambda} {\sigma_\Lambda} \leq \frac{1}{1- 2 \norm{h(\sigma_{A^c B^c})}_\infty} \un{i=1}{\ov{n}{\sum}} \left[  \entA {A_i} {\rho_\Lambda} {\sigma_\Lambda} + \entA {B_i} {\rho_\Lambda} {\sigma_\Lambda}   \right],
\end{equation}

In the third step of the proof, using the first two, we get a lower bound for the global MLSI constant of the whole chain in terms of the conditional MLSI constants on the aforementioned fixed-sized regions $A_i$ and $B_i$. For that, we need to consider that Assumption \ref{assump:1} holds true.

\begin{step}\label{step:3}
If Assumption \ref{assump:1} holds, we have:
\begin{equation*}\label{eq:step3}
\alpha (\LL_\Lambda^*) \geq \tilde{K} \un{i \in \qty{1, \ldots n}}{\text{min}} \qty{\alpha_\Lambda(\LL_{A_i}^*), \alpha_\Lambda(\LL_{B_i}^*)},
\end{equation*}
where $\ds \tilde{K}= \frac{1- 2K_1 e^{-K_2 l}}{2}$ and $\alpha_\Lambda(\LL_{A_i}^*)$, resp. $\alpha_\Lambda(\LL_{B_i}^*)$, denotes the conditional MLSI constant of $\LL_\Lambda^*$ on $A_i$, resp. $B_i$, as introduced in Definition \ref{def:conditionalMLSIconstant}.
\end{step}

\begin{proof}

By Equation \eqref{eq:step3-1} and Assumption \ref{assump:1}, we have
\begin{equation}\label{eq:step3-2}
\ent {\rho_\Lambda} {\sigma_\Lambda} \leq \frac{1}{1- 2 K_1 e^{-K_2 l}} \un{i=1}{\ov{n}{\sum}} \left[  \entA {A_i} {\rho_\Lambda} {\sigma_\Lambda} + \entA {B_i} {\rho_\Lambda} {\sigma_\Lambda}   \right].
\end{equation}

Now, by virtue of the definition of conditional MLSI constants on each $A_i$ and $B_i$, it is clear that
\begin{align*}
& \ent {\rho_\Lambda} {\sigma_\Lambda} \nonumber \\
&  \leq \frac{1}{1- 2  K_1 e^{- K_2 l}} \un{i=1}{\ov{n}{\sum}} \left[  \frac{-\tr[ \LL_{A_i}^* (\rho_\Lambda) \left(  \log \rho_\Lambda - \log \sigma_\Lambda \right) ]}{2 \alpha_\Lambda(\LL_{A_i^*})}+ \frac{-\tr[ \LL_{B_i}^* (\rho_\Lambda) \left(  \log \rho_\Lambda - \log \sigma_\Lambda \right) ]}{2 \alpha_\Lambda(\LL_{B_i^*})}  \right] \nonumber \\
&  \leq \frac{1}{1- 2 K_1 e^{- K_2 l}} \frac{1}{2 \un{i \in \qty{1, \ldots , n}}{\text{min}} \qty{\alpha_\Lambda(\LL_{A_i}^*), \alpha_\Lambda(\LL_{B_i}^*)}}   \un{i=1}{\ov{n}{\sum}} \left[  \text{EP}_{A_i}(\rho_\Lambda) + \text{EP}_{B_i}(\rho_\Lambda)  \right]\nonumber. 
\end{align*} 

Therefore,
\begin{align}\label{eq:step3-3}
& 2 \un{i \in \qty{1, \ldots, n}}{\text{min}} \qty{\alpha_\Lambda(\LL_{A_i}^*), \alpha_\Lambda(\LL_{B_i}^*)}\, \ent {\rho_\Lambda} {\sigma_\Lambda} \nonumber \\
&  \phantom{asdadsadads}  \leq \frac{1}{1- 2 K_1 e^{- K_2 l}} \left[ -\tr \left[ \left(\LL_\Lambda^*(\rho_\Lambda) + \LL_{A_n\cap B_n}^*(\rho_\Lambda) \right)  \left(\log \rho_\Lambda - \log \sigma_\Lambda \right) )  \phantom{\un{i=1}{\ov{n-1}{\sum}}} \right. \right. \nonumber  \\
& \phantom{asdadsadadsdadasasdasdasdads}  +  \left. \left.  \un{i=1}{\ov{n-1}{\sum}} \left(  \LL_{A_i \cap B_i}^*(\rho_\Lambda) + \LL_{A_{i+1} \cap B_i}^*(\rho_\Lambda)  \right) \left(\log \rho_\Lambda - \log \sigma_\Lambda \right)  \right]  \right] \nonumber  \\
& \phantom{asdadsadads}   \leq \frac{2}{1- 2  K_1 e^{-K_2 l}}  \left[- \tr[ \LL_\Lambda^*(\rho_\Lambda)  \left(\log \rho_\Lambda - \log \sigma_\Lambda \right)  ] \right],
\end{align} 
where we have used the locality of the Lindbladian and the positivity of the entropy productions.
  
Finally, note that the last term of expression (\ref{eq:step3-3}) is the entropy production of $\rho_\Lambda$. Hence, considering the quotient of this term over the relative entropy of the LHS, and taking infimum over $\rho_\Lambda \in \SSS_\Lambda$, we get
\begin{equation*}
\alpha (\LL_\Lambda^*)  = \un{\rho_\Lambda \in \SSS_\Lambda}{\text{inf}} \frac{\text{EP}(\rho_\Lambda)}{2 \ent {\rho_\Lambda} {\sigma_\Lambda}}  \geq \tilde{K} \un{i \in \qty{1, \ldots n}}{\text{min}} \qty{\alpha_\Lambda(\LL_{A_i}^*), \alpha_\Lambda(\LL_{B_i}^*)},
\end{equation*}
where $ \ds \tilde{K}:= \frac{1- 2 K_1 e^{-K_2 l }}{2}>0$.

\end{proof}

Finally, in the last step of the proof, we show that the conditional MLSI constants on every $A_i$ and $B_i$ are strictly positive and, additionally, independent of the size of $\Lambda$. For that, we need to suppose that Assumption \ref{assump:2} holds true. We also make use of some technical results from the previous section.

\begin{step}\label{step:4}
If Assumption \ref{assump:2} holds, for any $A_i$ defined as above we have
\begin{equation*}
\alpha_\Lambda \left(\LL_{A_i}^*\right)\geq C_{A_i}(\sigma_\Lambda) >0,
\end{equation*}
with $C_{A_i}(\sigma_\Lambda)$ independent of the size of $\Lambda$, and analogously for any $B_i$.

\end{step}

\begin{proof}
Consider $ X \in \qty{A_i , B_i \, : \, 1 \leq i \leq n  }$. Let us first recall that the conditional MLSI constant in $X$ is given by
\begin{align*}
\alpha_\Lambda (\LL_X^*) & = \un{\rla \in \SSS_\Lambda}{\text{inf}}  \frac{\text{EP}_X(\rho_\Lambda)}{2 \entA X \rla \sla}\\
& =  \un{\rla \in \SSS_\Lambda}{\text{inf}}  \frac{- \un{x \in X}{\sum}  \tr[\LL_x^*(\rla) \left( \log \rla - \log \sla  \right)  ] }{2 \entA X \rla \sla} .
\end{align*}

By virtue of Lemma \ref{lemma:EP>D}, we have
\begin{equation*}
\text{EP}_x(\rla) \geq \entA x \rla \sla
\end{equation*}
for every $x \in X$, and, thus,
\begin{equation}\label{eq:condMLSI>quotCRE}
\alpha_\Lambda (\LL_X^*) \geq \un{\rla \in \SSS_\Lambda}{\text{inf}}  \frac{ \un{x \in X}{\sum} \entA x \rla \sla}{2 \entA X \rla \sla}\,.
\end{equation}

Note that the quotient in the RHS of (\ref{eq:condMLSI>quotCRE}) is well-defined, since we have seen in Corollary \ref{cor:kernelA=kernelx} that the kernel of $\entA X \rla \sla$ coincides with the intersection of the kernels of $\entA x \rla \sla$ for every $x \in X$. Furthermore, because of Assumption \ref{assump:2}, we obtain the following lower bound for the conditional MLSI constant
\begin{equation}
\alpha_\Lambda (\LL_X^*) \geq \frac{1}{2 f_{X}(\sigma_{\Lambda})}\,,
\end{equation}
which is strictly positive, only depends on $\sigma_\Lambda$ and does depend on the size of $\Lambda$.

\end{proof}

Finally, putting together Steps \ref{step:1}, \ref{step:2}, \ref{step:3} and \ref{step:4}, we conclude the proof of Theorem \ref{thm:mainresult}.

\section{Mixing condition and strong quasi-factorization}\label{sec:assumptions}

\subsection{Mixing condition}\label{subsec:assump1}

In this subsection, we will elaborate on the mixing condition introduced in Assumption \ref{assump:1} and provide sufficient conditions for it to hold. Consider $\Lambda \subset \subset \Z$ a finite chain and  $A, B \subset \Lambda$  as in the splitting of $\Lambda$ in the proof of Theorem \ref{thm:mainresult} (see Figure \ref{fig:step1}). Denote $C:= B^c$ and $D:= A^c$, so that they can be expressed as the union of disjoint segments, $C = \un{i=1}{\ov{n}{\bigcup}} C_i$ and $D = \un{j=1}{\ov{n}{\bigcup}} D_j$, respectively. For every $i=1, \ldots, n$, \hl{respectively $i=1, \ldots , n-1$}, denote by $E_i$, resp. $F_i$, the connected set that separate $C_i$ from $D_i$, resp. $D_i$ from $C_{i+1}$ (see Figure \ref{fig:assump1}). Note that, because of the construction of $A$ and $B$ described in the previous section, every $E_i$ and $F_i$ are composed of \hl{$l$} sites, \hl{and every $C_i$ and $D_i$ of, at least, $2k-1$ sites}.

\begin{figure}
\begin{center}
\includegraphics[scale=0.45]{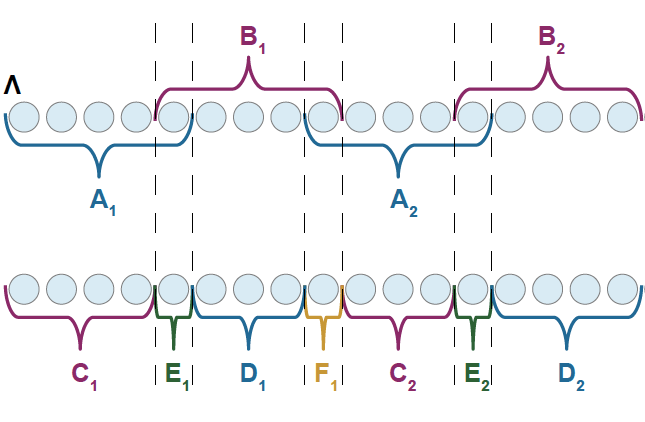}
\end{center}
\caption{Notation introduced in the splitting of $\Lambda$ into size-fixed $A_i$ and $B_i$ for the discussion in Assumption \ref{assump:1}. For simplicity we restrict to the case $k=2, l=1$. }
\label{fig:assump1}
\end{figure}

 Let  $\sla$ be the Gibbs state of a $k$-local commuting Hamiltonian. Then, with this construction, Assumption \ref{assump:1} can be read as the existence of positive constants $K_1, K_2$ independent of $\Lambda$ for which the following holds:

\begin{equation}\label{eq:Assump1CD}
\tag{A1}
\norm{\sigma_C^{-1/2} \otimes \sigma_D^{-1/2} \,  \sigma_{CD} \,  \sigma_C^{-1/2} \otimes \sigma_D^{-1/2} - \identity_{CD}}_\infty \leq K_1 e^{- K_2 l}, 
\end{equation}
where $l =d(C, D)$.

This exponential decay of correlations on the Gibbs state is similar to certain forms of decay of correlations of states that frequently appear in the literature of both classical and quantum spin systems. In the latter, this is closely related, for instance, to the concept of \textit{LTQO (Local Topological Quantum Order)}  \cite{[MP13]}, or the \textit{local indistinguishability} that was introduced in  \cite{[KB14]}.

The main difference with the (strong) \textit{mixing condition} of the classical case \cite{[D02]} lies in the fact that they considered a decay of correlations with the distance between two connected regions (in particular, rectangles), whereas in our case we have a finite union of regions of that kind. The fact that the regions are connected is essential for some properties that can be derived from the \textit{Dobrushin condition} (\cite{[DS87]} \cite{[Mar99]} \cite[Condition III.d]{[DS85]}). 

Nevertheless, the mixing condition that we need to assume for the proof of Theorem \ref{assump:1} to hold is actually a bit weaker. Indeed, the only necessary thing is that we can bound the LHS of (\ref{eq:Assump1CD}) by something that is strictly smaller than $1/2$, i.e., 
\begin{equation}\label{eq:Assump1CDweaker}
\tag{A1-weaker}
\norm{\sigma_C^{-1/2} \otimes \sigma_D^{-1/2} \,  \sigma_{CD} \,  \sigma_C^{-1/2} \otimes \sigma_D^{-1/2} - \identity_{CD}}_\infty < \frac{1}{2},
\end{equation}

It is clear that (\ref{eq:Assump1CD}) implies (\ref{eq:Assump1CDweaker}), as one can always choose $l$ big enough. This new condition is a bit more approachable and we will show below a couple of examples of systems satisfying it. First, we consider the situation in which the fixed point is close enough to the normalized identity (in other words, we consider high-enough temperature, depending logarithmically of the system size). Then, we can prove the following result.  

\begin{prop}\label{prop:HighTemperatureAssupmtion1}
Let $\Lambda$ be a finite chain and consider a splitting on it as the one of Figure \ref{fig:assump1}. Let us assume that $\sigma_\Lambda \in \SSS_\Lambda $ is close to the normalized identity in the following sense
\begin{equation}\label{eq:HighTemperature}
\norm{\sigma_\Lambda - \frac{\identity_\Lambda}{d_\Lambda}}_\infty < \varepsilon,
\end{equation}
for a certain small $\varepsilon \geq 0$. Then,  (\ref{eq:Assump1CDweaker}) holds.
\end{prop}

\begin{proof}
Note that Equation \eqref{eq:HighTemperature} is equivalent to
\begin{equation*}
(1- \varepsilon)\frac{\identity_\Lambda}{d_\Lambda}  < \sigma_\Lambda < (1 + \varepsilon) \frac{\identity_\Lambda}{d_\Lambda},
\end{equation*}
and, furthermore, given $C,D \subset \Lambda$ with $C \cap D = \emptyset$ (independently of the geometry) the following holds:
\begin{align*}
(1- \varepsilon)\frac{\identity_C}{d_C}  &< \sigma_C < (1 + \varepsilon) \frac{\identity_C}{d_C},
\end{align*}
as well as analogously for $D$ and $CD$, respectively. Recalling now that for every $A$, $B$ positive definite matrices the following equivalence holds: $A \leq B \Leftrightarrow A^{-1}\geq B^{-1}$, we have
\begin{equation*}
\frac{d_{CD}}{(1-\varepsilon)^2} \identity_{CD} > \sigma_C^{-1/2} \otimes \sigma_D^{-1/2} >  \frac{d_{CD}}{(1+\varepsilon)^2} \identity_{CD},
\end{equation*}
and thus
\begin{equation*}
\frac{(1-\varepsilon)}{(1+\varepsilon)^2} \identity_{CD} < \sigma_C^{-1/2} \otimes \sigma_D^{-1/2} \, \sigma_{CD} \, \sigma_C^{-1/2} \otimes \sigma_D^{-1/2}  <  \frac{(1+\varepsilon)}{(1-\varepsilon)^2} \identity_{CD}.
\end{equation*}

Therefore, it is easy to see that
\begin{equation*}
\norm{\sigma_C^{-1/2} \otimes \sigma_D^{-1/2} \, \sigma_{CD} \, \sigma_C^{-1/2} \otimes \sigma_D^{-1/2} - \identity_{CD}}_\infty < \text{max} \qty{\frac{3 \varepsilon + \varepsilon^2}{(1+\varepsilon)^2},\frac{3 \varepsilon - \varepsilon^2}{(1-\varepsilon)^2} },
\end{equation*}
where indeed this maximum is always attained by the second element, which constitutes a strictly decreasing function with $\varepsilon$. Hence, to conclude it is enough to find the values of $\varepsilon\geq 0$ for which 
\begin{equation*}
\frac{3 \varepsilon - \varepsilon^2}{(1-\varepsilon)^2} <\frac{1}{2},
\end{equation*}
which holds for $\varepsilon < 0.13$.  
\end{proof}

\begin{remark}
\hl{The previous result holds in particular for classical systems at high-enough temperature. Furthermore, note that the same proof allows to show that $(\ref{eq:Assump1CDweaker})$  holds for systems whose fixed point is close enough to a tensor product between $C$ and $D$ (with a distance scaling logarithmically with the system size). }
\end{remark}

Next, with a much more elaborate but similar in spirit proof, we can show that states with a defect at site $i$ so that the interaction is bigger there, but interactions decay away from that site, also satisfy (\ref{eq:Assump1CDweaker}).

\begin{prop}\label{prop:SpinDefect}
Let $\Lambda$ be a finite chain and consider a splitting on it as the one of Figure \ref{fig:assump1}. If we assume the following condition:
\begin{equation*}
\un{i=1}{\ov{n}{\prod}} \gamma_i^2 > \frac{2}{3} \un{i=1}{\ov{n}{\prod}} \delta_i^2 >\frac{1}{3},
\end{equation*}
where we are writing
\begin{itemize}
\item $\gamma_i:=\gamma^{(i)}_{CE} \, \gamma^{(i)}_{ED}\,  \gamma^{(i)}_{DF} \, \gamma^{(i)}_{FC}$, for $i=1, \ldots, n-1$, 
\item $\delta_i:=\delta^{(i)}_{CE}\,  \delta^{(i)}_{ED} \, \delta^{(i)}_{DF} \, \delta^{(i)}_{FC}$, for $i=1, \ldots, n-1$,
\item   $\ds \gamma_n:= \gamma_{CE}^{(n)} \gamma_{ED}^{(n)}$,
\item $\ds \delta_n:=\delta_{CE}^{(n)} \delta_{ED}^{(n)} $,
\end{itemize}
and for which each $\gamma^{(i)}_{GH}$, resp. $\delta^{(i)}_{GH}$, is the minimum, resp. maximum, eigenvalue of $\sigma_{(\partial G_i) \cap H_i}$,  then (\ref{eq:Assump1CDweaker}) holds.
\end{prop}

\begin{proof}

First, note that condition (\ref{eq:Assump1CDweaker}) is equivalent to the following

\begin{align}\label{eq:Assump1ineq}
\frac{1}{2}\sigma_C\otimes \sigma_D<\sigma_{CD}<\frac{3}{2}\sigma_C\otimes \sigma_D\,.
	\end{align}

\begin{figure}
\begin{center}
\includegraphics[scale=0.3]{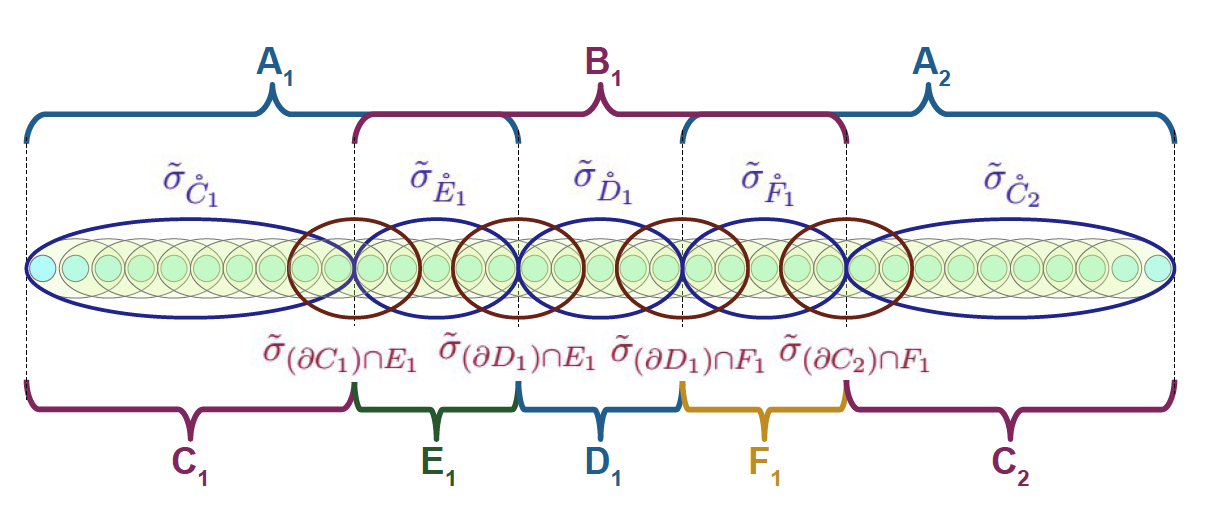}
\end{center}
\caption{Decomposition of $\sigma_\Lambda$ into the product of commuting terms for $k=3$ and $l=5$, assuming that $\Lambda$ is decomposed only into $A_1, B_1$ and $A_2$ for simplification.}
\label{fig:assump1explained}
\end{figure}
	
Now, the state $\sigma_\Lambda$ on the full chain can be decomposed into the following product of commuting terms (see Figure \ref{fig:assump1explained}):
\begin{equation}\label{eq:ExpressionSigma}
Z\sigma_{\Lambda}:= \left( \prod_{i=1}^{n-1}  \chi_i  \right) \tilde \sigma_{\mathring{C}_{n}} \tilde \sigma_{( \partial C_n) \cap E_n} \tilde \sigma_{\mathring{E_n}} \tilde \sigma_{( \partial D_n) \cap E_n }\tilde \sigma_{\mathring{D_n}},
\end{equation}
with
\begin{equation}\label{eq:ExpressionEachI}
\chi_i := \tilde \sigma_{\mathring{C}_i}\tilde \sigma_{( \partial C_i) \cap E_i} \tilde \sigma_{\mathring{E_i}}\tilde \sigma_{( \partial D_i)  \cap E_i}\tilde \sigma_{\mathring{D_i}}\tilde \sigma_{( \partial D_i) \cap F_i}\tilde \sigma_{\mathring{F_i}}\tilde \sigma_{( \partial C_{i+1}) \cap F_i }, 
\end{equation}
and where $Z$ is the normalization factor and  $\mathring{G}$ denotes the interior of $G$, that is the the set of sites in $G$ whose corresponding interaction is fully supported in $G$. We use the notation $\tilde \sigma_G$ to remark that this term does not coincide in general with $\tr_{G^c}[\sigma_\Lambda]$.  We will bound the boundary terms as follows: For any consecutive $G_j , H_i \in\{C_j, D_i, E_i , F_i \}$ so that $H_i= E_i$ or $F_i$ (and thus  $G_j=C_i$, $C_{i+1}$ or $D_i$), we have
\begin{align}\label{eq:BoundSigmaBoundary}
\gamma_{GH}^{(i)}\,\identity_{( \partial G_j ) \cap H_i }\le \tilde	\sigma_{( \partial G_j ) \cap H_i } \le\,\delta_{GH}^{(i)}\,\identity_{( \partial G_j ) \cap H_i }\,.
\end{align}

Note that, in a slight abuse of notation, we are denoting by $\gamma^{(i)}_{FC}$ and $\delta^{(i)}_{FC}$ the coefficients corresponding to the term $\tilde \sigma_{F_i \cap ( \partial C_{i+1})}$.Then, since $( \partial G_j ) \cap H_i$ consists of $2(k-1)$ sites, half of which belong to $G_j$ and the other half to $H_i$, we can write
\begin{equation*}
\gamma_{GH}^{(i)}\,\tilde \sigma_{\mathring{G}_i} \otimes  \tilde \sigma_{\mathring{H}_i}  \le \,\tilde \sigma_{\mathring{G}_i} \tilde \sigma_{(\partial G_i) \cap H_i} \tilde \sigma_{\mathring{H}_i} \le\,\delta_{GH}^{(i)}\,\tilde \sigma_{\mathring{G}_i} \otimes  \tilde \sigma_{\mathring{H}_i} 
\end{equation*}
and thus replacing (\ref{eq:BoundSigmaBoundary}) in (\ref{eq:ExpressionSigma}) after tracing out $E$ and $F$, it is easy to show that 
\begin{align*}
  \left( \prod_{i=1}^{n-1}\gamma_{i}\,\tilde \sigma_{\mathring{C}_i\mathring{D}_i} \tr(\tilde \sigma_{\mathring{E}_i\mathring{F}_i}) \right) \gamma_{CE}^{(n)} \gamma_{ED}^{(n)} \, & \tilde  \sigma_{\mathring{C}_n \mathring{D}_n} \tr(\tilde \sigma_{\mathring{E}_n})  \, \\
	& \le  Z \,\sigma_{CD} \\
	& \le 	\left( \prod_{i=1}^{n-1}\delta_{i}\,\tilde \sigma_{\mathring{C}_i\mathring{D}_i} \tr(\tilde \sigma_{\mathring{E}_i\mathring{F}_i}) \right) \delta_{CE}^{(n)} \delta_{ED}^{(n)} \, \tilde  \sigma_{\mathring{C}_n \mathring{D}_n} \tr(\tilde \sigma_{\mathring{E}_n})\,,
\end{align*}	
where $\gamma_i:=\gamma^{(i)}_{CE}\gamma^{(i)}_{ED}\gamma^{(i)}_{DF}\gamma^{(i)}_{FC}$  and  $\delta_i:=\delta^{(i)}_{CE}\delta^{(i)}_{ED}\delta^{(i)}_{DF}\delta^{(i)}_{FC}$, and $d$ is the dimension of the local Hilbert space associated to each site.  

On the other hand, if we proceed analogously to get a bound for $\sigma_C \otimes \sigma_D$ to compare it with $\sigma_{CD}$, we obtain
\begin{align*}
   \left( \prod_{i=1}^{n-1}\gamma_i^2\,\tilde \sigma_{\mathring{C}_i}\tilde \sigma_{\mathring{D}_i}\tr(\tilde \sigma_{\mathring{E}_i\mathring{F}_i})^2  \right) & \left(\gamma_{CE}^{(n)} \gamma_{ED}^{(n)}\right)^2 \tilde \sigma_{\mathring{C}_n } \tilde \sigma_{\mathring{D}_n } \tr(\tilde \sigma_{\mathring{C}\mathring{D}})\\
&\le Z^2	\sigma_C\otimes \sigma_D\\
& \le  \left( \prod_{i=1}^{n-1}\delta_i^2\,\tilde \sigma_{\mathring{C}_i}\tilde \sigma_{\mathring{D}_i}\tr(\tilde \sigma_{\mathring{E}_i\mathring{F}_i})^2 \right) \left(\delta_{CE}^{(n)} \delta_{ED}^{(n)}\right)^2 \tilde \sigma_{\mathring{C}_n} \tilde \sigma_{\mathring{D}_n }  \tr(\tilde \sigma_{\mathring{C}\mathring{D}})\,.
\end{align*}	

Therefore, a sufficient condition for (\ref{eq:Assump1CDweaker}) is that
\begin{align}\label{sufficientcond}
\frac{1}{2}\tr(\sigma_{\mathring{C}\mathring{D}})\prod_{i=1}^{n-1}\frac{\gamma_i\tr(\sigma_{\mathring{E}_i\mathring{F}_i})}{\delta_i^2}	\frac{\gamma_n}{\delta_n^2}< Z <\frac{3}{2}\tr(\sigma_{\mathring{C}\mathring{D}})\prod_{i=1}^{n-1}\frac{\gamma_i^2\tr(\sigma_{\mathring{E}_i\mathring{F}_i})}{\delta_i}\frac{\gamma_n^2}{\delta_n} \,.
\end{align}	
with $\ds \gamma_n:= \gamma_{CE}^{(n)} \gamma_{ED}^{(n)}$ and $\ds \delta_n:=\delta_{CE}^{(n)} \delta_{ED}^{(n)} $.
Note that, when $\beta\to0$,  $Z\to d^{\abs{\Lambda}}$, where the number $|\Lambda|$ of sites is equal to $|\mathring{E}|+|\mathring{F}| +8(k-1)(n-1)+|\mathring{C}|+|\mathring{D}|$. Moreover, $\delta_i=\gamma_i=1$ in the limit. Therefore (\ref{sufficientcond}) holds trivially, since it reduces to $\ds \frac{1}{2}<1<\frac{3}{2} $. It is reasonable then to think that, close  to infinite temperature \hl{(in a distance depending logrithmically with the system size)}, (\ref{sufficientcond}) holds.

Indeed, let us assume the following inequality between the $\gamma_i$ and $\delta_i$,
\begin{equation}\label{eq:SuffCondAssump1Weaker}
 \un{i=1}{\ov{n}{\prod}} \gamma_i^2 > \frac{2}{3}   \un{i=1}{\ov{n}{\prod}} \delta_i^2 >\frac{1}{3}.
\end{equation}

To conclude the proof that equation \eqref{eq:SuffCondAssump1Weaker} implies equation \eqref{eq:Assump1CDweaker}, it is enough to bound $Z$, the normalization factor, in the same way that we have bounded $\sigma_{CD}$ and $\sigma_C \otimes \sigma_D$. Introducing those bounds in the inequalities appearing in  \eqref{sufficientcond}, it is easy to see that this expression reduces to \eqref{eq:SuffCondAssump1Weaker}.

\end{proof}

\subsection{Strong quasi-factorization}\label{subsec:assump2}

In this subsection, we will discuss Assumption \ref{assump:2}, which can be seen as a strong quasi-factorization of the relative entropy, and provide some sufficient conditions on $\sigma_\Lambda$ for it to hold.

Given $\Lambda$ a finite chain and $A$ a subset of $\Lambda$, if we denote by $\sla$ the Gibbs state of a $k$-local commuting Hamiltonian, Assumption \ref{assump:2} reads as:
\begin{equation}\label{eq:Assump2}
\entA A \rla \sla \leq f_A(\sigma_\Lambda) \un{x \in A}{\sum} \entA x \rla \sla \; \; \; \forall \rho_\Lambda \in \SSS_\Lambda, 
\end{equation}
where $1 \leq f_A(\sigma_\Lambda) < \infty $ depends only on $\sigma_\Lambda$ and is independent of $\abs{\Lambda}$.

Let us first recall that $A$ has a fixed size of $2(k+l)-1$ sites, so $\abs{ A \partial} = 2 (2k + l -1) -1$, and is, in particular, fixed. Moreover, if we separate one site from the rest in each step, i.e., for every $2 \leq m \leq \abs{A}$, if we consider the only connected $B^{(m)} \in A$ of size $m$ that contains the first site of $A$, and we split $B^{(m)}$ into two connected regions $B^{(m)}_1$ and $B^{(m)}_2$ so that $\abs{B^{(m)}_1} = 1$, it is clear that the following inequality 
\begin{equation}\label{eq:Assump2ind}
\entA {B^{(m)}} \rla \sla \leq f_{B^{(m)}}(\sigma_\Lambda) \left[  \entA {B_1^{(m)}} \rla \sla +  \entA {B_2^{(m)}} \rla \sla \right] \; \; \; \forall \rho_\Lambda \in \SSS_\Lambda\,, 
\end{equation}
implies inequality \eqref{eq:Assump2} by induction, taking
\begin{equation*}
f_A(\sigma_\Lambda) := \un{2 \leq m \leq \abs{A}}{\text{sup}} f_{B^{(m)}}(\sigma_\Lambda)\,.
\end{equation*}

Therefore, we can pose the following natural question.

\begin{question}
Given two adjacent subsets $A, B \subset\Lambda$, under which condition on the Gibbs state $\sigma_\Lambda$ does there exist a bounded $f_{AB}(\sigma_\Lambda)$ only depending on  $\sigma_\Lambda$ and independent of the size of $\Lambda$  such that the following inequality holds for every $\rho_\Lambda \in \SSS_\Lambda$:
\begin{equation}\label{conjecture}
\entA {AB} {\rho_\Lambda} {\sigma_\Lambda} \leq f_{AB}(\sigma_\Lambda) \left( \entA {A} {\rho_\Lambda} {\sigma_\Lambda} + \entA {B} {\rho_\Lambda} {\sigma_\Lambda}  \right) \, ?
\end{equation}

\end{question}

Remark that we only need to answer this question for $\abs{A}, \abs{B} < 2(k+l)$. Although we cannot give a general answer to this question, we can provide some motivation for situations in which it might hold. For that, we prove before the following lemma, which shows that a conditional relative entropy in a certain region can be upper bounded by a quantity depending only on the reduced states in that region independently of the cardinality of the whole lattice. 

\begin{lemma}\label{lemma:boundCRE}
	Let $A \subset \Lambda$.  For any $\rho_\Lambda\in\SSS_\Lambda$,
	$$D_A(\rho_\Lambda\|\sigma_\Lambda)\le  D_A(\rho_\Lambda\|\sigma_A\otimes \sigma_{A^c})+D(\rho_{A\partial}\|\sigma_{A \partial})\,.$$
\end{lemma}	
\begin{proof}
	A simple use of the definition of the conditional relative entropy leads to the following identity:
	\begin{align}\label{eq8}
		D_A(\rho_\Lambda\|\sigma_\Lambda)- D_A(\rho_\Lambda\|\sigma_A\otimes \sigma_{A^c}) & = \ent {\rla} {\sla} - \ent {\rla} {\sigma_A\otimes \sigma_{A^c}} \nonumber \\
		&=\tr[ \rho_\Lambda\left(   -\log\sigma_\Lambda+\log\sigma_A\otimes\sigma_{A^c} \right) ] \,.
	\end{align}	
	
	By the quantum Markov chain property of the state $\sigma_\Lambda$ between $A \leftrightarrow \partial A \leftrightarrow \left( A \partial \right)^c$ and by Proposition \ref{prop:identityQMC}, we have
	\[\log\sigma_{\Lambda}=\log\sigma_{A^c}+\log\sigma_{A\partial}-\log\sigma_{\partial A}\,.\]

Plugging this in Equation \eqref{eq8} we arrive at:
	\begin{align*}\label{eq88}
		D_A(\rho_\Lambda\|\sigma_\Lambda)- D_A(\rho_\Lambda\|\sigma_A\otimes \sigma_{A^c})
		& =\tr[ \rho_\Lambda\left(   -\log\sigma_{A\partial}+\log\sigma_A\otimes\sigma_{\partial A} \right) ] \\
		& = \ent{\rho_{A\partial}}{\sigma_{A\partial}}-\ent{\rho_{A\partial}}{\sigma_{A}\otimes\sigma_{\partial A}} \\
		& \leq \ent{\rho_{A\partial}}{\sigma_{A\partial}}\,.
	\end{align*}
	\end{proof}
	
Note that, when $\rho$ is classical,  inequality \eqref{conjecture} holds true for any Gibbs state of a classical $k$-local commuting Hamiltonian in 1D, and under some further assumptions it also does in more general dimensions, since \eqref{conjecture} coincides in the classical setting with a usual result of quasi-factorization of the entropy, due to the DLR conditions. More specifically, this inequality holds classically whenever the Dobrushin-Shlosman complete analiticity condition holds. Moreover, in that setting one can see that $f_{AB}(\sigma_\Lambda)$ actually depends only on $\sigma_{(AB)\partial}$.

 It is then reasonable to believe that this inequality might also hold true for  Gibbs states of quantum $k$-local commuting Hamiltonians in 1D, although $f_{AB}$ could possibly depend on $\sigma$ on the whole lattice $\Lambda$ (without depending on its size). The intuition behind this is that $\sigma_\Lambda$ is also a quantum Markov chain, and Lemma \ref{lemma:boundCRE} shows that the conditional relative entropy in a certain region can be approximated by its analogue for $\sigma_\Lambda$ a tensor product obtaining an additive error term that can be bounded by something that only depends on the region and its boundary. 
 
 However, if we define 
 \begin{equation*}
f_{AB}(\sigma_\Lambda) := \un{\rho_\Lambda \in \SSS_\Lambda}{\text{sup}} \frac{\entA {AB} {\rho_\Lambda} {\sigma_\Lambda} }{\entA {A} {\rho_\Lambda} {\sigma_\Lambda} + \entA {B} {\rho_\Lambda} {\sigma_\Lambda} }
\end{equation*}
we lack a proof that, in general, it satisfies the necessary conditions for \eqref{conjecture} to hold. The study of examples of Hamiltonians whose Gibbs state satisfies the aforementioned inequality is left for future work.

Nevertheless, let us  recall here some situations for which we already know that inequality \eqref{conjecture} holds. First, if $\sigma_\Lambda$ is a tensor product, this inequality holds with $f=1$ \cite{[CLP18]}, as a consequence of strong subadditivity. Moreover,  for a more general $\sigma_\Lambda$, if $A$ and $B$ are separated enough, we have seen in Step \ref{step:2} that it also holds with  $f=1$, due to the structure of quantum Markov chain of $\sigma_\Lambda$. Since in \eqref{conjecture} we are assuming that $A$ and $B$ are adjacent, we cannot use this property to ``separate'' $A$ from $B$, i.e. write $\sigma_\Lambda$ as a direct sum of tensor products that separate $A$ from $B$, and thus the proof of Step \ref{step:2} cannot be used here.

\hl{Additionally, let us mention that the idea used in Proposition $\text{\ref{prop:HighTemperatureAssupmtion1}}$ to show that Assumption $\text{\ref{assump:1}}$ holds for systems at high-enough temperature cannot be used for Assumption $\text{\ref{assump:2}}$. Indeed, by assuming Equation $\text{\eqref{eq:HighTemperature}}$, we get }
\begin{equation*}
\highlight{\log (1-\varepsilon) + \log \frac{\identity_X}{d_X}	<\log \sigma_X < \log (1+\varepsilon) + \log \frac{\identity_X}{d_X} \, ,	 }
\end{equation*} 
\hl{for $X= A, B, \Lambda$, and thus }
\begin{equation*}
\highlight{ D_{AB} \left(\rho_\Lambda \Big\| \frac{\identity_\Lambda}{d_\Lambda} \right) + \log \left( \frac{1-\varepsilon}{1+\varepsilon} \right)<D_{AB} (\rho_\Lambda || \sigma_\Lambda) <  D_{AB} \left(\rho_\Lambda \Big\| \frac{\identity_\Lambda}{d_\Lambda} \right) + \log \left( \frac{1+\varepsilon}{1-\varepsilon} \right) \, , }
\end{equation*}
\hl{and analogously for $D_{A} (\rho_\Lambda || \sigma_\Lambda)$ and $D_{B} (\rho_\Lambda || \sigma_\Lambda)$. Therefore, this allows to reduce the expression above for $f_{AB}(\sigma_\Lambda)$  to }
\begin{equation*}
\highlight{ f_{AB}(\sigma_\Lambda) > \un{\rho_\Lambda \in \SSS_\Lambda}{\text{sup}} \; \frac{D_{AB} \left(\rho_\Lambda \Big\| \frac{\identity_\Lambda}{d_\Lambda} \right) }{D_{A} \left(\rho_\Lambda \Big\| \frac{\identity_\Lambda}{d_\Lambda} \right) + D_{B} \left(\rho_\Lambda \Big\| \frac{\identity_\Lambda}{d_\Lambda} \right) } + \frac{ g(\varepsilon) }{D_{A} \left(\rho_\Lambda \Big\| \frac{\identity_\Lambda}{d_\Lambda} \right) + D_{B} \left(\rho_\Lambda \Big\| \frac{\identity_\Lambda}{d_\Lambda} \right) } \, , }
\end{equation*}
\hl{where $g(\varepsilon)$ is a function of $\varepsilon$, and even though the first summand in the right-hand side above is known to be lower bounded by $1$, the second term can be in principle arbitrarily large. Therefore, this method does not help to prove Assumption $\text{\ref{assump:2}}$, with a multiplicative error term, for high-enough temperature. However, it would if we allowed for an additive term to appear. Indeed, using the previous comparisons for the logarithms it is easy to show that the following inequality holds:}
\begin{equation*}
\highlight{D_{AB}(\rho_\Lambda || \sigma_\Lambda)  < D_A (\rho_\Lambda || \sigma_\Lambda) + D_B (\rho_\Lambda || \sigma_\Lambda) + 3 \log \left( \frac{1+ \varepsilon}{1- \varepsilon} \right) \, , }
\end{equation*}
\hl{using the fact that $f_{AB}(\identity_\Lambda /d_\Lambda) =1$. This inequality, though, cannot be used to show the positivity of a MLSI constant, although some generalizations of quasi-factorization inequalities with additive error terms for a different notion of conditional relative entropy have found different applications in quantum information theory, including new entropic uncertainty relations, as shown in $\text{\cite{BardetCapelRouze-ApproximateTensorization-2020}}$. }

\section{Conclusions and open problems}\label{sec:conclusions}

In this paper, we have addressed the problem of finding conditions on the Gibbs state of a local commuting Hamiltonian so that the generator of a certain dissipative evolution has a positive modified logarithmic Sobolev constant. Building on results from classical spin systems \cite{[D02]} \cite{[Cesi01]}, and following the steps of \cite{[KB14]}, where the authors addressed the analogous problem for the spectral gap, we have developed a strategy based on five points that provides positivity of the modified logarithmic Sobolev constant. Moreover, we have used this strategy to present two conditions on a Gibbs state so that its corresponding heat-bath dynamics in 1D satisfies this positivity. 

This strategy opens a new way to obtain positivity of MLSI constants, and thus it will probably be of use to prove this condition not only for the dynamics studied in this paper, but for some other dynamics such as the one of the Davies semigroups, for instance \cite{Davies1979}. However, for the time being, some natural questions arise from this work, such as the existence of examples of non-trivial Gibbs states for which these static conditions, and thus positivity of the modified logarithmic Sobolev constant, hold. 

\textbf{\textit{Question 1.}} \textit{Are there any easy examples of $\sigma_\Lambda$ for which the strong quasi-factorization of Assumption \ref{assump:2} holds with $f$ different from $1$?}

So far, the only example we have for this condition to hold is for $\sigma_\Lambda$ a tensor product everywhere, for which the value of $f$ is always $1$. It is reasonable to think that this condition holds, for instance, when $\sigma_\Lambda$ is a classical state, since in this case one could expect to recover the classical case, in which this inequality would agree with the usual quasi-factorization thanks to the DLR condition. However, this is left for future work.

\textbf{\textit{Question 2.}} \textit{Are there any more examples of $\sigma_\Lambda$ for which the mixing condition of Assumption \ref{assump:1} holds?}

Even though we have mentioned that this condition holds for classical states and we have shown a more complicated example of Gibbs state verifying this in Proposition \ref{prop:SpinDefect}, most of the tools available in the setting of quantum many-body systems to address the problem of decay of correlations on the Gibbs state depend strongly on the geometry used to split the lattice, and more specifically on the number of boundaries between the different regions $A$ and $B$. Since, in our case, this number scales linearly with $\Lambda$, there is no hope to use any of those tools to obtain more examples of Gibbs states satisfying Assumption \ref{assump:1}. However, it is possible that a different approach allows for more freedom in this sense. 

\textbf{\textit{Question 3.}} \textit{Could this result be extended to more dimensions?}

Following the same approach from this paper, we would need to cover a $n$-dimensional lattice with small rectangles overlapping pairwise in an analogous way to the construction described here for dimension 1. It is easy to realize that, even in dimension 2, one would need at least three systems to classify the small rectangles so that two belonging to the same class would not overlap. Thus, for our strategy to hold in dimension, at least, 2, we would need a result of quasi-factorization that provides an upper bound for the relative entropy of two states in terms of the sum of three conditional relative entropies, instead of two, and a multiplicative error term. Since we are lacking a result of this kind so far, this question constitutes an open problem.

\textbf{\textit{Question 4.}} \textit{Can we change the geometry used to split the lattice?}

Another possible approach to tackle this problem could be based on the geometry presented in the classical papers \cite{[D02]}, \cite{[Cesi01]} and the quantum case for the spectral gap \cite{[KB14]}. In this approach, in each step one splits the rectangle into two connected regions and carries out a more evolved geometric recursive argument. One of the main benefits from this approach would be a weakening in the mixing condition assumed in the Gibbs state. However, the main counterpart would be the necessity of a strong result of quasi-factorization for the relative entropy, even stronger than the one appearing in \eqref{conjecture} since the multiplicative error term should converge to $1$ exponentially fast, in which both sides of the inequality would contain conditional relative entropies.

\section*{Acknowledgements}
The authors would like to thank Nilanjana Datta for fruitful discussions and for
her comments on an earlier version of the draft. IB is supported by French
A.N.R. grant: ANR-14-CE25-0003 "StoQ". AC was partially supported by a La
Caixa-Severo Ochoa grant (ICMAT Severo Ochoa project SEV-2011-0087, MINECO), by an MCQST Distinguished Postdoc, and
AC and DPG acknowledge support from MINECO (grant MTM2017-88385-P) and from
Comunidad de Madrid (grant QUITEMAD-CM, ref. P2018/TCS-4342). AL acknowledges
support from the Walter Burke Institute for Theoretical Physics in the form of
the Sherman Fairchild Fellowship as well as support from the Institute for
Quantum Information and Matter (IQIM), an NSF Physics Frontiers Center (NFS
Grant PHY-1733907), and from the BBVA Fundation and the Spanish 'Ramón y Cajal' Programme (RYC2019-026475-I / AEI / 10.13039/501100011033). This project has received
funding from the European Research Council (ERC) under the European Union’s
Horizon 2020 research and innovation programme (grant agreement No 648913). CR
acknowledges financial support from the TUM university Foundation Fellowship, and AC and CR also acknowledge financial support
by the DFG cluster of excellence 2111 (Munich Center for Quantum Science and
Technology).

\bibliographystyle{abbrv}
\bibliography{library}

\appendix

\end{document}